\documentclass[journal]{IEEEtran}
\usepackage{cite}
\ifCLASSINFOpdf
  \usepackage[pdftex]{graphicx}
\else
\fi
\usepackage{amsmath}
\usepackage{amsfonts}
\usepackage{amssymb}
\usepackage{color}
\usepackage{algorithmic}
\usepackage{algorithm}
\usepackage{array}
\ifCLASSOPTIONcompsoc
  \usepackage[caption=false,font=normalsize,labelfont=sf,textfont=sf]{subfig}
\else
  \usepackage[caption=false,font=footnotesize]{subfig}
\fi
\usepackage{fixltx2e}
\usepackage{stfloats}
\usepackage{url}
\newcommand{\norm}[1]{\left\lVert#1\right\rVert}	
\usepackage{amsmath, amsthm, amssymb}
\newtheorem{corollary}{Corollary}
\newtheorem{proposition}{Proposition}
\newtheorem{remark}{Remark}
\renewcommand{\qedsymbol}{$\blacksquare$}
\usepackage{booktabs}
\usepackage{varwidth}

\begin{document}

\title{A Novel NOMA Solution with RIS Partitioning}

\author{Aymen~Khaleel,~\IEEEmembership{Graduate Student Member,~IEEE} and
        Ertugrul~Basar,~\IEEEmembership{Senior Member,~IEEE}
\thanks{This work was supported by the Scientific and Technological Research Council of Turkey (TUBITAK) under Grant 120E401.
	
The authors are with the Communications Research and Innovation Laboratory (CoreLab),  Department of Electrical and Electronics Engineering, Ko\c{c} University, Sariyer 36050, Istanbul, Turkey. \mbox{Email: akhaleel18@ku.edu.tr, ebasar@ku.edu.tr}}% <-this % stops a space
}

\maketitle

\begin{abstract}
 Reconfigurable intelligent surface (RIS) empowered communications with non-orthogonal multiple access (NOMA) has recently become an appealing research direction for next-generation wireless communications. In this paper, we propose a novel NOMA solution with RIS partitioning, where we aim to enhance the spectrum efficiency by improving the ergodic rate of all users, and to maximize the user fairness. In the proposed system, we distribute the physical resources among users such that the base station (BS) and RIS are dedicated to serve different clusters of users. Furthermore, we formulate an RIS partitioning optimization problem to slice the RIS elements between the users such that the user fairness is maximized. The formulated problem is shown to be a non-convex and non-linear integer programming (NLIP) problem with a combinatorial feasible set, which is challenging to solve. Therefore, we exploit the structure of the problem to bound its feasible set and obtain a sub-optimal solution by sequentially applying three efficient search algorithms. Furthermore, we derive exact and asymptotic expressions for the outage probability. Simulation results clearly indicate the superiority of the proposed system over the considered benchmark systems in terms of ergodic sum-rate, outage probability, and user fairness performance.
\end{abstract}


% \fi
\IEEEpeerreviewmaketitle

\section{Introduction}
\IEEEPARstart{N}{on-orthogonal}  multiple
%%%%%%%%%%%%%%%%%%N O M A%%%%%%%%%%%%%%
 access (NOMA) has been regarded as one of the promising technologies to address the increasing demand for high data rates, massive connectivity, and spectrum efficiency associated with fifth-generation (5G) and beyond networks. Due to the use of higher carrier frequencies and in order to fulfill the verticals' specific requirements, local service areas are among the main envisioned deployment models in 5G \cite{5G-local}. Within this context, NOMA can play an important role in order to efficiently exploit the spectrum and serve higher numbers of users in such networks \cite{NOMA-5G}. This is because NOMA allows the sharing of the same time/frequency/code resources between users and thus, enhances the spectrum efficiency and decreases the latency by allowing more users to be connected in the same time-frequency slot \cite{NOMA}. In power domain (PD)-NOMA \cite{PD-NOMA}, the difference in channel gains of different users is exploited and, accordingly, different power levels are assigned to different users by the use of superposition coding (SC) at the base station (BS) side. At the receiver side, each user employs the successive interference cancellation (SIC) technique to recover its own message. Compared to the orthogonal multiple access (OMA), it has been shown that NOMA has a superior performance in terms of the outage probability and the ergodic sum-rate \cite{NOMA-OP}, \cite{NOMA-Frame},\cite{NOMA-OMA}.
 
%%%%%%%%%%%%%%% R I S %%%%%%%%%%%%%%%%%%%
Recently, reconfigurable intelligent surface (RIS)-assisted communication has received growing attention as a potential next-generation technology due to its promising capabilities to control the wireless propagation environment. RIS is an array of low-cost and passive reflecting elements that can be used to re-engineer the electromagnetic waves by adjusting the reflection coefficient of each element \cite{RIS1}, \cite{RIS2}. Due to their promising advantages, RISs have been integrated to many existing wireless technologies as a main or supportive module in the communications system. In \cite{RIS-Space}, two RIS-assisted space shift keying (SSK) schemes are proposed to enhance the performance of the classical SSK system in terms of the bit error rate and system throughput. In \cite{RIS-Quad}, the authors proposed RIS-based receive quadrature reflecting modulation by partitioning the surface into two parts, where the in-phase and quadrature components are sent from each part separately. In \cite{EB1}, an RIS is integrated to a conventional single-input single-output (SISO) system to achieve an ultra-reliable communication system. In \cite{EB2}, an RIS is used as a part of the transmitter to perform different types of  index modulation (IM) and thus, obviating the need for a multi-antenna BS. These techniques are extended to multiple-input multiple-output (MIMO) systems in \cite{AK}, where an RIS is used to replace the  radio frequency (RF) chains in Alamouti's scheme, and to boost the data rate for vertical Bell Labs layered space-time (VBLAST) system by using IM.
%%%%%%%%%%%%%% Related Works: RIS - N O  M A%%%%%%%%%%%%
\vspace{-0.2cm}
\subsection{Related Works}
 Recently, the joint operation of RISs with NOMA has appeared as an appealing research direction and the combination of the hardware capabilities of RISs and the SC technique of PD-NOMA has been investigated. In \cite{Benchmark} and \cite{Bench2}, the rate performance and user fairness of an RIS-assisted NOMA system are optimized by maximizing the minimum decoding signal-to-interference-plus-noise ratio (SINR) of all users, which is achieved by the joint optimization of the active and passive beamforming at the BS and the RIS, respectively. The authors in \cite{RIS-NOMA-Vinc} considered an RIS-assisted  multiple-input single-output (MISO) NOMA system where the cell-center users are served by using spatial division multiple access (SDMA) and the cell-edge users are served by RISs. In \cite{RIS_NOMA_Net}, by considering a priority-oriented design, an RIS-assisted SISO NOMA network is proposed, where the passive beamforming weights are designed at the RIS side in order to enhance the spectrum efficiency. In \cite{RIS-NOMA-ENRGY}, the energy efficiency is enhanced in a downlink RIS-assisted NOMA system by jointly optimizing the user clustering, passive beamforming, and power allocation. The user fairness is considered in \cite{RIS-NOMA-Fair}, where the authors investigated the joint optimization of power allocation, decoding order, and the RIS phase shifts to maximize the minimum user rate considering a total power constraint. In \cite{RIS-NOMA-FRAME}, the deployment and passive beamforming design of an RIS are investigated for an RIS-assisted MISO  NOMA system, in order to maximize the energy efficiency under the constraint of preserving the individual data rate requirements for users. Joint optimization for the active beamforming matrices at the BS and the reflection coefficient vector at the RIS is utilized in \cite{RIS-NOMA-MultiCls} in order to minimize the total transmit power for a multi-cluster MISO NOMA networks. In \cite{RIS-NOMA-SC}, a multi-cluster RIS-assisted MIMO NOMA network is considered, where by designing its passive beamforming weights, the RIS is employed in a signal cancellation mode to eliminate the inter-cluster interference. The authors in \cite{RIS-NOMA-SG1} and \cite{RIS-NOMA-SG2} used stochastic geometry to investigate the coverage probability and ergodic rate of an RIS-assisted multi-cell NOMA networks for outdoor scenarios by using Poisson cluster process (PCP) model. The authors in \cite{RIS-NOMA-Cord} considered the combination of joint transmission coordinated multipoint (JT-CoMP) with the RIS technology in order to enhance the cell-edge user ergodic rate performance without degrading the performance of the cell-center user. In \cite{RIS-NOMA-Phs}, the authors investigated the impact of the coherent and the random discrete phase-shifting designs on an RIS-assisted NOMA system. Finally, the authors in \cite{RIS-NOMA-ResAloc} considered the resource allocation problem in an RIS-assisted NOMA system, and jointly optimized the channel assignments, power allocation, decoding order, and RIS reflection coefficients, in order to maximize the system throughput.
 \vspace{-0.5cm}
\subsection{Motivation and Contributions} 
In light of the above discussion, it can be noted that a common physical resource (PR) allocation scheme is followed by all of the previous works, which can be summarized as follows. The BS and RIS are both used to serve all users by using a single SC message, where all users are assumed to be in the field-of-view (FoV) of the RIS. Furthermore, by adjusting its reflection coefficients, the RIS is used as a single unit to serve all users jointly. Considering this PR allocation scheme, the main factor that limits the users’ performance becomes the mutual interference that underlies the SC technique, which can be seen clearly when the users are deployed randomly in and out of the FoV of the RIS. In such a users' deployment scenario, not all users share the same PR (BS and RIS) and therefore, the use of a single SC message for all users adversely affects user fairness and unnecessarily amplifies the mutual interference between the users' messages.  

Against this background, we propose an RIS-assisted novel NOMA system, where a more efficient PR allocation scheme is proposed to enhance user fairness and effectively mitigate the impact of the mutual interference between users. In the proposed system, the users are grouped into two clusters, Cluster 1 ($\text{C}_1$) contains all the users out of the FoV of the RIS, and  Cluster 2 ($\text{C}_2$) contains the ones inside it.
The BS is dedicated to serve the users in $\text{C}_1$ and the RIS is portioned into sub-surfaces, where each sub-surface is exploited to serve a different user in $\text{C}_2$. The main contributions of this paper can be summarized as follows:
\vspace{-0.40cm}
\begin{itemize}[
	\setlength{\IEEElabelindent}{\dimexpr-\labelwidth-\labelsep}% Wrapping of text beyond first line of \item
	\setlength{\itemindent}{\dimexpr\labelwidth+\labelsep}% identation for each new \item
	\setlength{\listparindent}{\parindent}]
	\item To the best of the authors' knowledge, this study considers the PR problem when the users are deployed in and out of the FoV of the RIS, for the first time. To address this issue, we propose a novel PR scheme that employs  RIS partitioning to mitigate the mutual interference between the users in and out of the FoV of the RIS. This leads to an effective enhancement in the performance of all users in terms of ergodic rate, outage probability, fair distribution of the PRs among users, and simplifies the detection process.
    \item We formulate an RIS partitioning optimization problem to find a proper RIS slicing that maximizes the fairness among $\text{C}_2$ users. Although the fact that the formulated problem is a non-convex and non-linear integer programming (NLIP) one, we exploit the structure of the problem, specifically the nature of transmission over the RIS, to provide a sub-optimal solution with marginal performance degradation. Note that, unlike the partitioning schemes adopted in \cite{RIS-part1} and \cite{RIS-part2}, to obtain a scalable optimization framework for the phase-shift adjustment of large RISs, we partition the RIS so that each sub-surface works as a modulator and beamformer simultaneously. In this way, each sub-surface sends an independent data stream to the user assigned to it independently from the other sub-surfaces.
    \item We derive the exact and asymptotic outage probability expressions for users in $\text{C}_2$. Accordingly, under the uniform partitioning scenario, we obtain the required number of RIS elements that need to be allocated for each user in $\text{C}_2$ to obtain a given outage probability value for all users.
    \item With comprehensive computer simulations, we compare our proposed system with four different benchmark schemes and show the superiority of our proposed system in terms of the ergodic sum-rate, outage probability, and user fairness.
    \end{itemize}
\indent The rest of the paper is organized as follows. In Section \ref{sec:Main}, we introduce the system model and describe the transmission mechanism in detail. The outage probability and its asymptotic behaviour are formulated in Section \ref{sec:OP}. In Section \ref{sec:SpltAprh}, we introduce our proposed RIS partitioning approach and provide the system performance analysis of two special cases for the proposed system in different deployment scenarios. Computer simulations are provided in Section \ref{sec:Simu} followed by the conclusions in Section \ref{sec:Concl}.\footnote{\textit{Notation}: Matrices and column vectors are denoted by an upper and lower case boldface letters, respectively. $\mathbf{X}\in\mathbb{C}^{m\times k}$ denotes a  complex-valued matrix $\mathbf{X}$ with $m\times k$ size, where $\mathbf{X}^T$ is the transpose and $[\mathbf{X}]_{n,\tilde{n}}$ is the ($n,\tilde{n}$)-th entry. $\mathbf{0}_N$, $\mathbf{I}_N$, $m\choose k$, $\lfloor\cdot\rfloor$, $\lceil\cdot\rceil$,  and $\text{mod}(\cdot)$ are the $N$-dimensional all-zeros column vector, the $N\times N$ identity matrix, the binomial coefficient, the floor, ceiling, and modulus functions, respectively. $x\sim\mathcal{CN}(0,\sigma^2)$ stands for complex Gaussian distributed random variable (RV) with mean $\text{E}[x]=0$ and variance $\text{VAR}[x]=\sigma^2$. $\mathbb{R}$, $\mathbb{Z}^+$, and $|\mathcal{S}|$ are the set of real numbers, the set of positive integer numbers, and the cardinality of the set  $\mathcal{S}$, respectively.}
\vspace{-0.4cm}
%%%%%%%%%%%%End of Literature%%%%%%%%%
%%%%%%%%%%%%%%%% S Y S T E M M O D E L%%%%%%%%%%%%%%%%%%%%%
\section{RIS Partitioning and Rotating: System Model}\label{sec:Main}
 Consider a downlink NOMA system where single-antenna users are served by a single-antenna BS\footnote{The RIS is deployed in the direct line-of-sight (LoS) of the BS to compensate for the high path loss associated with the RIS. This leads to a rank-one BS-RIS channel, where no multiplexing gain can be achieved by using multiple transmit antennas. Furthermore, it has been shown that the array gain vanishes as the number of users increases irrespective of the number of transmit antennas and RIS size, and the number of users that can be efficiently served is one \cite{max-min}.} and an RIS of $N=N_HN_V$ elements, where $N_H$ and $N_V$ denote the number of elements per row and per column, respectively. The users are assumed to be grouped into two clusters\footnote{$\text{C}_1$ and $\text{C}_2$ are the main clusters that identify the users who are out of the FoV of the RIS and served by the BS, and the ones who are in the FoV of the RIS and served by the RIS, respectively. Therefore, it is possible to use a second level of clustering inside these two main clusters \cite{NOMA-Clustering}, however, the investigation of such a scenario is out of the scope of this study.}, $\text{C}_1$ with $ M_1$ users and $\text{C}_2$ with $M_2$ users, as shown in Fig. 1. The RIS is deployed close, yet in the far-field\footnote{This simplifies the system analysis, where for the the near-field assumption we need to consider the distances of the individual RIS elements from the BS and thus, the effective BS antenna area and the polarization mismatch associated with each RIS element \cite{near-field}. }, to the BS in order to exploit the pure line-of-sight (LoS) channels and to make the use of a backhaul link from the BS to the RIS practical. Furthermore, denoting the $m^{th}$ user in $\text{C}_2$ by $\text{U}_{m,2}$, perfect channel state information (CSI) of the BS-RIS-$\text{U}_{m, 2}$ link for all users in $\text{C}_2$ and the channel gains of all users in $\text{C}_1$ need to be available at the BS side, under the quasi-static flat-fading channels assumption. Although the fact that the perfect CSI assumption is practically challenging, nevertheless, it is adopted by the vast majority of the RIS-assisted NOMA works in the literature \cite{max-min}, \cite{csi_perfect1}, \cite{csi_perfect2}. Therefore, the CSI is assumed to be obtained by using one of the proposed methods in the literature \cite{ch-est1}, \cite{ch-est2}, and thus, the performance results provided in this study can serve as an upper bound to the ones achievable in practical implementation.
 
 Instead of using a single SC message that combines all users' symbols, the $\text{C}_1$ and $\text{C}_2$ users' symbols are simultaneously transmitted over the BS direct link and the RIS reflection link, respectively, as follows.
 %%%%%%%%%%%%%Main Figure %%%%%%%%%%%%%%
 \begin{figure}[t]
 	%\begin{center}
 	\includegraphics[width=91mm, height=50mm]{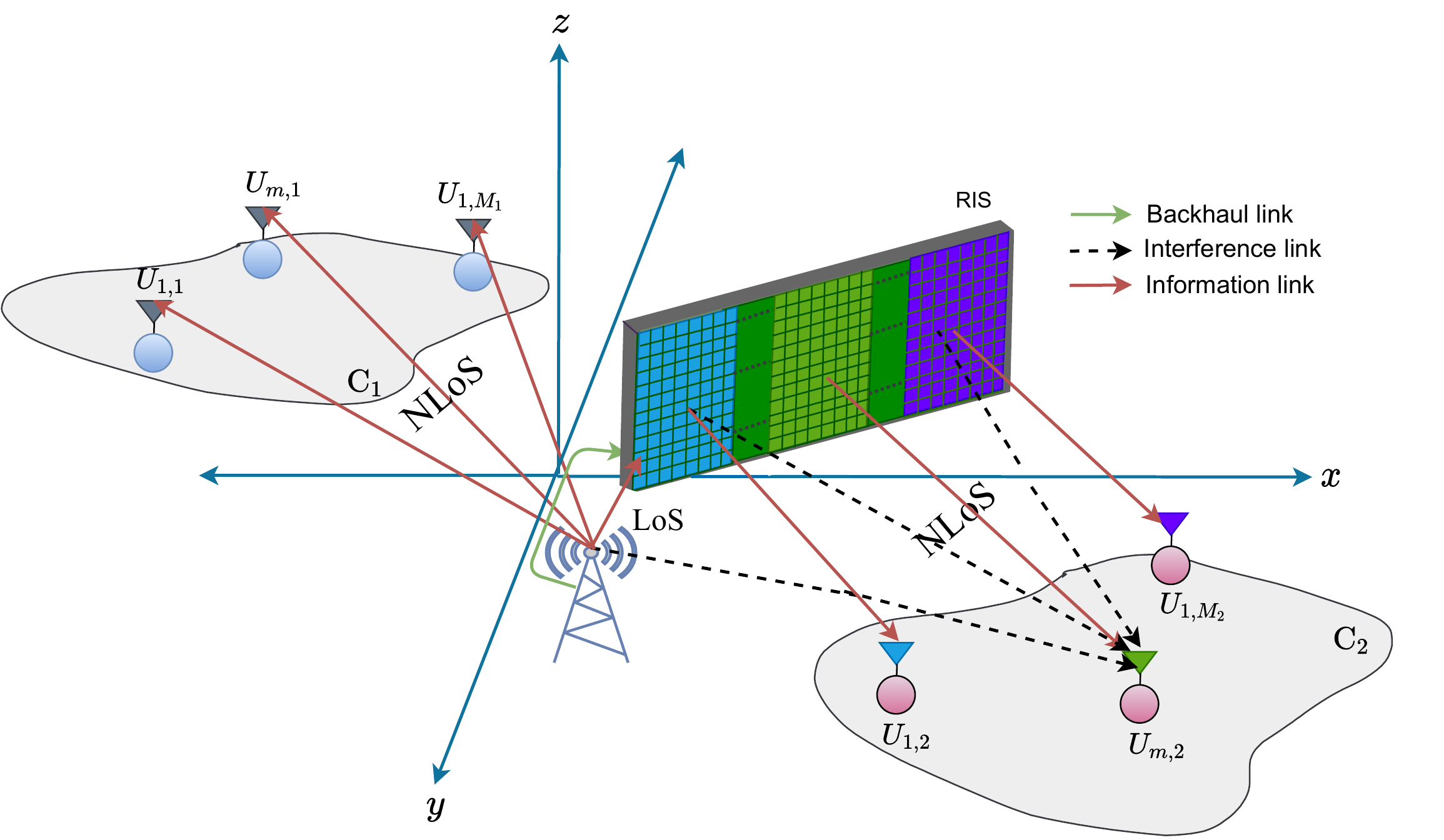}
 	\caption{RIS partitioning based NOMA system.}\label{fig:BlockDiagram}
 	%\end{center}
 	\vspace{-0.5cm}
 \end{figure}
\vspace{-0.0cm}
 %%%%%%%%%%%%%%%%%%%%%%%%%%%%%%%%% 
%%%%%%%%%%%%%%%%%%%%%C_1 Transmission%%%%%%%%%%%%%% 
 \subsection{The transmission of $\text{C}_1$ users' symbols}
Let $T_c$ denotes the duration of the channel coherence block, wherein all channels remain constant, while all channel coefficients are independent and identically distributed (i.i.d.) over different coherence blocks. Within $T_c$, the BS transmits $x$, which is the superposed signal of all symbols to be transmitted to $M_1$ users in $\text{C}_1$. As in conventional PD-NOMA, $x$ is constructed as follows:
\setlength{\abovedisplayskip}{3pt}
\setlength{\belowdisplayskip}{3pt}
\begin{align}
	x=&\sum_{m=1}^{M_1}\sqrt{P\zeta_{m}}x_{m},\label{eq:x}
\end{align}
where $P$ is the transmit power, $\zeta_{m}$ and $x_{m}$ are the power allocation factor and the symbol to be transmitted to user $m$ in $\text{C}_1$
($\text{U}_{m, 1}$), respectively, where $\sum_{m=1}^{M_1}\zeta_{m}=1$. Note that $x$ can be represented in the form $x=ue^{-j\theta_{M_1}}$, where $u=|x|$ is the amplitude and $\theta_{M_1}=\text{arg}(x)$ is the angle. From \eqref{eq:x}, it is clear that $u$ and $\theta_{M_1}$ are both RVs, where $\theta_{M_1} \in[0,2\pi)$ and $u\in\mathbb{R}$. Without loss of generality, in order to simplify our derivations, $u$ is assumed to be a constant value that is equal to unity, which can be achieved by the proper design of the constellations used to obtain $x_1$ and $x_2$ for $\text{C}_1$ users \cite{Constellations1}, \cite{Constellations2}. For example, consider the simple case of two users in $\text{C}_1$, where $x_1\in\{\frac{t}{\sqrt{2}}(1+1j), \frac{t}{\sqrt{2}}(-1-1j)\}$ and $x_2\in\{\frac{t}{\sqrt{2}}(1-1j), \frac{t}{\sqrt{2}}(-1+1j)\}$, then for any random values of $\zeta_1$ and $\zeta_2$ we obtain $u=|t|$, where $t\in\mathbb{R}$.

The signal received by each user $m$ in $\text{C}_1$ is given by
\begin{align}
	\tilde{y}_{m}=\tilde{v}_{m}x+\tilde{z}_{m},\label{eq:y1}
\end{align}
where $\tilde{v}_m$ is the BS-$\text{U}_{m, 1}$ channel coefficient and $\tilde{z}_{m}$ is the additive white Gaussian noise (AWGN) sample at $\text{U}_{m, 1}$. 

From \eqref{eq:y1}, it can be noted that there is no interference from $\text{C}_2$ users’ symbols received by $\text{C}_1$ users, due to the location of  $\text{C}_1$ being behind the RIS. Although $\text{C}_1$ users still need to use SIC to recover their own signals, the SIC is performed with significantly less number of iterations due to the absence of $\text{C}_2$ users’ interference. This simplifies the detection process and effectively reduces the error propagation associated with the SIC technique. In this way, the system performance for the $M_1$ users in $\text{C}_1$ follows the one of the conventional PD-NOMA, therefore, its analysis is omitted in this study and we focus on the system performance of $\text{C}_2$. Nevertheless, we still perform the numerical simulations in Section V for the users in $\text{C}_1$ along with $\text{C}_2$ users to show that the proposed system enhances the performance of both $\text{C}_1$ and $\text{C}_2$ users compared to the benchmark systems.
\vspace{-0.0cm}
%%%%%%%%%%%%%%%%%%%%%%%%%%%%%%%%%%%%%%%%%%%%%%%%
\subsection{The transmission of $\text{C}_2$ users' symbols}
Instead of using the SC technique to send the symbols of $\text{C}_2$ users in a single message from the BS, the symbol of each user $m$ is sent over the RIS reflection link independent from the symbols belong to the other users. Therefore, the RIS is partitioned into $M_2$ sub-surfaces, where each sub-surface $i$ has $N_i$ elements and is allocated to serve a specific user $m$ in $\text{C}_2$, for $i=m$, where $m=1, 2, ... M_2$. Specifically, each sub-surface $i=m$ remodulates the same impinging signal $x$ to reflect the $\text{U}_{m,2}$ symbol over the RIS-$\text{U}_{m,2}$ reflection link. Thus, by considering only a single reflection from the RIS elements \cite{RIS-Reflct} and within the same $T_c$, the signal received by each user $m$ in $\text{C}_2$ can be obtained as follows: 
\begin{align}
	y_{m}&\hspace{-0.1cm}=\hspace{-0.1cm}\left[\mathbf{g}_{m,m}^T\mathbf{\Theta}_{m}\mathbf{h}_m\hspace{-0.1cm}+\hspace{-0.1cm}\smash[b]{\underset{i\neq m}{\sum}}_{i=1}^{M_2-1} \mathbf{g}_{i,m}^T\mathbf{\Theta}_{i}\mathbf{h}_i+v_m \right]x\hspace{-0.1cm}+\hspace{-0.1cm}z_m \label{eqSplit1},
\end{align}
where $z_m\sim\mathcal{CN}(0,\sigma^2)$ is the AWGN sample at $\text{U}_{m, 2}$, $v_m$ is the BS-$\text{U}_{m, 2}$ channel coefficient, $v_m\sim\mathcal{CN}(0,L^{\text{BS}}_m)$, under Rayleigh fading assumption \cite{max-min}, \cite{LoS-Rayleigh},  where $L^{\text{BS}}_m$ denotes the BS-$\text{U}_{m, 2}$ path gain. Assuming pure LoS links \cite{LoS-Rayleigh}, $\mathbf{h}_i\in\mathbb{C}^{N_i\times 1}$ is the BS-($i^{th}$ sub-surface) LoS channel vector with $[\mathbf{h}_i]_n=\sqrt{L^{\text{RISh}}}e^{-j\psi^{(n)}}$, where $L^{\text{RISh}}$ denote the path gain, $\psi^{(n)}=\pi(n-1) \text{sin}\psi_A\text{sin}\psi_E$,  where $\psi_E$ and $\psi_A$ denote the LoS elevation and azimuth angles of arrival (AoA)s at the RIS, respectively \cite{max-min}. Since the RIS is deployed in the far-field of the BS, the path gain experienced by each element is assumed to be the same, where the spacing between the elements is small compared to the BS-RIS distance. $\mathbf{g}_{i,m}\in\mathbb{C}^{N_i\times 1}$ is the RIS ($i^{th}$ sub-surface)-$\text{U}_{m, 2}$ channel vector, $[\mathbf{g}_{i,m}]_n=\sqrt{L^{\text{RISg}}_{m}}g_{i,m}^{(n)}$, where $g_{i,m}^{(n)}$ and  $L^{\text{RISg}}_{m}$ are the RIS ($i^{th}$ sub-surface $n^{th}$ element)-$\text{U}_{m, 2}$ small-scale fading coefficient and path gain, respectively. $g_{i,m}^{(n)}=\beta_{i,m}^{(n)}e^{-j\phi_{i,m}^{(n)}}$, where $\beta_{i,m}^{(n)}$ and $\phi_{i,m}^{(n)}$ denote the channel amplitude and phase, respectively. Thus, the overall RIS-$\text{U}_{m, 2}$ channel vector can be given as $\mathbf{g}_m=[\mathbf{g}_{1,m}\;...\;\mathbf{g}_{i,m}\;...\;\mathbf{g}_{M_2,m}]^T$, $\mathbf{g}_m\in\mathcal{CN}(\mathbf{0}_N,\mathbf{R}_{\text{RIS}})$, where $\mathbf{R}_{\text{RIS}}\in \mathbb{C}^{N\times N}$ is the RIS spatial correlation matrix. $\mathbf{\Theta}_i\in\mathbb{C}^{N_i\times N_i}$ is the matrix of reflection coefficients for the $i^{th}$ sub-surface of the RIS, $[\mathbf{\Theta}_i]_n=\eta_i^{(n)}e^{j\Phi_i^{(n)}}$, where $\Phi_i^{(n)}\in [0,2\pi)$ and $\eta_i^{(n)}=1$, $\forall i, n$, assuming full reflection\footnote{Note that, practically, there are phase-dependent amplitude variations associated with the reflection coefficient of each RIS element \cite{phase-dependent}. However, we assume a constant reflection amplitude for the mathematical analysis simplification, as this assumption has no effect on the proposed system design.}. In what follows, we describe the process of remodulating the impinging signal $x$ at the RIS sub-surfaces to reflect the symbols of the $\text{C}_2$ users.

By properly and independently adjusting its phase shifts, each sub-surface $i=m$ performs two roles, namely, maximizing the channel gain by making the BS-RIS-$\text{U}_{m, 2}$ channel phases equal to zero (passive beamforming), and remodulating the impinging signal $x$ to reflect the symbol to be transmitted to $\text{U}_{m,2}$. Note that since the transmission of the RIS to $\text{C}_2$ depends on the BS transmission to $\text{C}_1$, the symbol rate is the same for all users in both clusters. To serve user $m$ in $\text{C}_2$, the phases of the $i=m$ sub-surface are adjusted such that the phase of each element is given as
%%%%%%%%%%%%%Phase adjustment %%%%%%%%%%%%%%%%%
\begin{align}
	\Phi_m^{(n)}=\phi_{m,m}^{(n)}+\psi^{(n)}+\theta_m+\theta_{M_1}, n=1, ..., N_m,\label{eq:phs-adjust}
\end{align}
%%%%%%%%%%%%%Phase adjustment %%%%%%%%%%%%%%%%%
where $\theta_m$ is the phase shift keying (PSK) symbol to be transmitted to user $m$ in $\text{C}_2$, furthermore, $\Phi_m^{(n)}$ is assumed to be calculated at the BS and sent to the RIS controller over a backhaul link. By considering the close distance and the LoS link between the BS and the RIS, a wired/out-band wireless backhaul link can be used without affecting the useful bandwidth used in \eqref{eq:x} \cite{feedback-link1}, \cite{backhaul2}. In this way, the BS and RIS perform a joint transmission synchronized by the backhaul link as in coordinated multi-point transmission \cite{coordinate}, where the achievable capacity of the backhaul link is assumed to be much higher than the one of the RIS-$\text{C}_2$ link.

The PSK modulation scheme is adopted for the users in $\text{C}_2$ due to the phase-dependent amplitude variation associated with the RIS reflection coefficient, where it is difficult to realize modulation schemes with a non-constant envelope \cite{phase-dependent}. 

\indent According to the RIS phase adjustment in \eqref{eq:phs-adjust}, \eqref{eqSplit1} can be re-expressed as 
\vspace{-0.1cm}
\begin{align}
	y_m=&\sqrt{L_m^{\text{RIS}}}\left[e^{j\theta_m}\sum_{n=1}^{N_m}\beta_{m,m}^{(n)}+\smash[b]{\underset{i\neq m}{\sum}}_{i=1}^{M_2-1} \left[\sum_{n=1}^{N_i}\beta_{i,m}^{(n)}e^{j\bar{\Phi}_i^{(n)}}\right]\right]\nonumber\\
	&+e^{j\theta_{M_1}}v_m+z_m,\label{eq:y22}
\end{align}
where $L_m^{\text{RIS}}=L_m^{\text{RISh}}L_m^{\text{RISg}}$,$\bar{\Phi}_i^{(n)}=-\phi_{i,m}^{(n)}+\Phi_i^{(n)}$, $\Phi_i^{(n)}$ corresponds to the phase adjustment of the other sub-surface $i$ for $i\in\{1, ..., M_2\}\setminus\{m\}$, and it is given by
\vspace{-0.0cm}
\begin{align}
\Phi_i^{(n)}=&\phi_{i,i}^{(n)}+\psi^{(n)}+\theta_i+\theta_{M_1}, n=1, ..., N_i.
\end{align}
The first three terms in \eqref{eq:y22} represent the amplitudes of the constructive combining, the intra-cluster, and inter-clusters interference, respectively. Furthermore, thanks to the remodulation process at the RIS, it can be noted that there is no interference from $\text{C}_1$ users' symbols received by $\text{C}_2$ users over the RIS link.\\
\indent In order to detect its PSK symbol, each user $m$ in $\text{C}_2$ is assumed to perform maximum likelihood (ML) detection in the presence of the interference coming from the other sub-surfaces ($i\neq m$) and the one comes from the BS, which are irremovable \cite{CRadio} and considered as noise. In this way, for the detection process, user $m$ in $\text{C}_2$ needs to know the overall sum of the channel gains of the $m^{th}$ sub-surface only. This, in turn, significantly reduces the training overhead and limits the spectrum efficiency degradation associated with it \cite{ch-est1}, \cite{ch-est2}. Furthermore, by properly choosing the RIS size \cite{BS-Link-Ignored}, each user $m$ in $\text{C}_2$ relies on the amplification gain provided by its own sub-surface ($i=m$) to overcome the irremovable interference. Thus, the SINR for user $m$ in $\text{C}_2$ to decoded its own message is given by 
\vspace{-0.2cm}
%%%%%%%%%%%%%%%instantaneous rate 2%%%%%
\begin{align}
	%%%%numerator%%%%%%%%%%%
	\text{SINR}_m&=\frac{A_m}
	%%%%Denomerator%%%
	{I_m+\frac{1}{\rho}},\label{eq:sinr2}
\end{align}
%%%%%%%%%%%%%%%instantaneous rate 2%%%%% 
%%%%%%%%%%%%%%%instantaneous rate 1%%%%%
where $\rho=\frac{P}{\sigma^2}$ denotes the transmit SNR, $A_m$ and $I_m$ denote the signal and total interference powers, respectively, and they are, from \eqref{eq:y22}, given as
\vspace{-0.2cm}
\begin{align}
	A_m&=\left|\sqrt{L_m^{\text{RIS}}}\sum_{n=1}^{N_m}\beta_{m,m}^{(n)}\right|^2, \label{eq:A1}\\
	I_m&=\left|I_{\text{RIS}}+v_m\right|^2,\label{eq:I1}
\end{align}
and $I_\text{RIS}$ is given by
\vspace{-0.2cm}
\begin{align}
I_{\text{RIS}}=\sqrt{L_m^{\text{RIS}}}\smash[b]{\underset{i\neq m}{\sum}}_{i=1}^{M_2-1} \left[\sum_{n=1}^{N_i}\beta_{i,m}^{(n)}e^{j\bar{\Phi}_i^{(n)}}\right]\label{eq:I_RIS}.
\end{align}
%%%%%%%%%%%%%%%instantaneous rate 2%%%%% 
From (\ref{eq:sinr2}), the instantaneous transmission rate for $\text{U}_{m, 2}$ and the sum-rate for the all $M_2$ users can be calculated, respectively, as follows:
\begin{align}
	R_m=&\log_2(1+\text{SINR}_m), \label{eq:R}
\end{align}
\begin{align}
%%%%%%%%%%%%Sum rate%%%%%%%%
	R=&\sum_{m=1}^{M_2}R_m=\sum_{m=1}^{M_2}\log_2(1+\text{SINR}_m).
\end{align}
In the following section we derive the outage probability for a given user $m$ in $\text{C}_2$.
\vspace{-0.0cm}
\begin{remark}
It can be verified from \eqref{eq:y1} and \eqref{eq:y22} that the number of users in $\text{C}_1$ and $\text{C}_2$ has no effect on the performance of the users belong to the other cluster, where the BS and RIS independently serve $\text{C}_1$ and $\text{C}_2$, respectively. Furthermore, although the fact that the users in $\text{C}_2$ still receive interference from the users in $\text{C}_1$ through the BS-$\text{U}_m$ link in addition to the sub-surfaces mutual interference, increasing the RIS size can effectively mitigate the impact of the overall interference on the users in $\text{C}_2$ \cite{BS-Link-Ignored}. Furthermore, the sum-rate gain provided by the proposed system does not require that the users in $\text{C}_2$ have different channel gains as is the case in conventional NOMA, which adds more flexibility to the system design. Finally, for the detection process, user $m$ in $\text{C}_2$ needs to know the overall sum of the channel gains of the $m^{th}$ sub-surface only, which effectively reduces the training overhead in the channel estimation of the RIS channels.
\end{remark}
\vspace{-0.5cm}
%%%%%%%%%End of System Model%%%%%%%%%%%
%%%%%%% Outage Probability%%%%%%%%%%%%%
\section{Outage Probability Analysis}\label{sec:OP}
Denoting the data rate requirement for $\text{U}_{m, 2}$ as $\gamma^*_m$, the outage probability for $\text{U}_{m, 2}$ is given as \cite{proakis}
\begin{align}\label{eq:P}
P^{out}_m&=P(R_m<\gamma_m^*)
\end{align}
By substituting (\ref{eq:R}) in (\ref{eq:P}), we obtain
\begin{align}
P^{out}_m&=P\left(\log_2\left(1+\text{SINR}_m\right)<{\gamma_m^*}\right),\nonumber\\
&=P\left(\text{SINR}_m<2^{\gamma_m^*}-1\right)\label{eq:OP-SNR}
\end{align}
Due to the spatial correlation between the RIS-$\text{U}_m$ channels, and thus, the correlation through $\mathbf{\Theta}_i$, it is challenging to derive the distribution of $\text{SINR}_m$\cite{op_deter}, \cite{op_match}\footnote{In \cite{op_deter} and \cite{op_match} a deterministic equivalent and moment matching approaches have been considered, respectively, to obtain the outage probability under spatially correlated channels. However, these works do not consider the intelligent phase shift adjustment.}. Furthermore, by considering an inter-element separation of $\lambda/2$, where $\lambda$ is the wavelength of the operating frequency, the spatially correlated channels become close to the i.i.d. case \cite{Corr}. This approximation is verified through Monte Carlo simulations in Section \ref{sec:Simu}, where the ergodic-rates and outage probability curves with/without correlation are shown to be close to each other. In what follows the outage probability expression is given under the i.i.d Rayleigh fading channels assumption and thus, the obtained results can serve as an upper bound for the performance of the cases where the inter-element separation is less than $\lambda/2$ and the RIS channels are highly spatially-correlated.
%%%%%%%%%%%%%%%%%%%%%%%%%%%%%%%%%%%%%
%%%%%%%%%% Proposition 1 %%%%%%%%%%%%%%%%%
\begin{proposition}\label{prop:Proposition 1}
The closed-form outage probability expression of user $m$ in $\text{C}_2$, assuming uncorrelated channels $\mathbf{R}_{\text{RIS}}=\mathbf{I}_N$, is given by
\begin{align}
	P^{\text{out}}_m&=1-\texttt{\textbf{\textit{Q}}}_{\frac{1}{2}}\left(\frac{\mu_1}{s_1},\sqrt{\frac{y}{s_1^2}}\right)+\left(\frac{s_2^2}{s_1^2+s_2^2}\right)^{\frac{1}{2}}\text{exp}\left(\frac{y}{2s_2^2}\right)\nonumber\\
	&\hspace{-0.2cm}	\times\text{exp}\left(-\frac{\mu_1^2}{2(s_1^2+s_2^2)}\right) 
	\texttt{\textbf{\textit{Q}}}_{\frac{1}{2}}\left(\frac{\mu_1}{s_1}\sqrt{\frac{s_2^2}{s_1^2+s_2^2}},\sqrt{\frac{y(s_1^2+s_2^2)}{s_1^2s_2^2}}\right),\label{eq:OP1}
\end{align}
where $\texttt{\textbf{\textit{Q}}}_k$ is the $k^{th}$ order generalized Marcum Q-function \cite{marcum}, $y=\frac{2^{\gamma_m^*}-1}{\rho}$,  $\mu_1=\sqrt{L_m^{\text{RIS}}}N_m\frac{\sqrt{\pi}}{2}$, $s_1^2=L^{\text{RIS}}_{m}N_m\frac{4-\pi}{4}$, and $s_2^2=0.5(2^{\gamma_m^*}-1)(L^{\text{RIS}}_{m}(N-N_m)+L_m^{\text{BS}})$.
\end{proposition}
\begin{proof}
	See Appendix A.
\end{proof}
Proposition 1 expresses the outage probability as a function of the transmit SNR and the ratio of $\sqrt{L_m^{\text{RIS}}}N_m$ to $\sqrt{L_m^{\text{RIS}}}(N-N_m)$ plus $\sqrt{L_m^{\text{BS}}}$, where these parameters reflect the amplification gain, the sub-surfaces interference, and the BS interference powers, respectively. To get more insight, we give the following corollaries that follow from Proposition 1.
\begin{corollary}
The Asymptotic behaviour of the outage probability can be obtained, for $\rho\rightarrow \infty$ or $y\rightarrow 0$, as
\begin{align}
	P_m^{\infty}&=\left(\frac{s_2^2}{s_1^2+s_2^2}\right)^{\frac{1}{2}}\text{exp}\left(-\frac{\mu_1^2}{2(s_1^2+s_2^2)}\right), \label{eq:OP_c1}
\end{align} 
\end{corollary}
\begin{proof}
The proof follows directly from Proposition 1 by letting $y=0$, where $	\texttt{\textbf{\textit{Q}}}_{\frac{1}{2}}(a,0)=1, \forall a$ \cite{marcum}.
\end{proof}
%\vspace{-0.2cm}
From Corollary 1, at high SNR, the outage probability performance improves exponentially with the ratio of the amplification gain to the overall interference associated with the RIS sub-surfaces and the BS links.
\vspace{-0.25cm}
\begin{corollary}
Consider the case where the RIS is uniformly partitioned between $M_2+1$ users in $\text{C}_2$, where $N_m=N_i=N/(M_2+1), \forall i, m$, $\sqrt{L_m^{\text{RIS}}}M_2N_m>>L_m^{BS}$, and $\rho\rightarrow \infty$, then, the outage probability is given by
\begin{align}
	P_m^{\infty}&\hspace{-0.1cm}=\hspace{-0.1cm}\left(\frac{2M_2}{2M_2+4-\pi}\right)^{\frac{1}{2}}\hspace{-0.1cm}\text{exp}\left(-\frac{\pi N_m}{2(2M_2+4-\pi)}\right),\label{eq:OP_c2} 
\end{align}	
for $M_2>>1$, \eqref{eq:OP_c2} can be simplified to
\begin{align}
P_m^{\infty}&\approx\text{exp}\left(-\frac{\pi N_m}{4M_2}\right)\label{eq:OP_c22}.
\end{align}
Thus, for a given $P_m^{\infty}$, the required $N_m$ for each user in $\text{C}_2$ can be obtained as
\begin{align}
N_m\approx\lceil -\frac{4M_2}{\pi}\ln(P_m^{\infty})\rceil.\label{eq:N}
\end{align}
\end{corollary}
\begin{proof}
The proof follows directly from Corollary 1 by letting $s_2^2=0.5(\sqrt{L_m^{\text{RIS}}}(N-N_m)+L_m^{BS})\approx 0.5\sqrt{L_m^{\text{RIS}}}M_2N_m$, where, without loss of generality, $\gamma^*_m=1$ bit per channel use (bpcu). 
\end{proof}
\vspace{-0.4cm}
Corollary 2 shows the outage probability performance at high SNR by considering only the mutual sub-surfaces interference impact. The BS link interference impact can be ignored by assuming large RIS size $N$ and/or $L_m^{\text{BS}}>>L_m^{\text{RIS}}$. It can clearly be seen from \eqref{eq:OP_c22} that the outage probability performance improves exponentially in proportion to the ratio of the sub-surface size allocated to user $m$ to the number of the other served users ($M_2$). Although increasing $N_m$ for each user, and hence increasing $N$, increases the number of the interferer signals received by each user, the outage probability still decreases exponentially for all users. This can be explained by the fact that the interference is of the incoherent type where the interferer signals are combined constructively/destructively in a random way. 
\vspace{-0.2cm}
\begin{remark}
Interestingly, the impact of the mutual sub-surfaces' interference between users can be effectively mitigated by increasing the overall RIS size for a given $M_2$ number of users in $\text{C}_2$. This is in contrast to the case of using SC, where increasing the transmit power has no effect on the mutual interference between the superposed symbols. Furthermore, it is worth noting that the outage probability in \eqref{eq:OP1} can be generalized to the case where all the links, BS-RIS, BS-$\text{U}_m$, and RIS-$\text{U}_m$ are Rician fading channels. In this case, by considering the same derivation steps provided in Appendix A, we obtain $Y$ to be the difference of two independent and non-central chi-square random variables, which has the following characteristic function (CF)\cite{cdf}
\begin{align}
	\Psi_{Y}(w)=\frac{\text{exp}\left(\frac{jw\mu_{A}^2}{1-2jw\sigma_{A}^2}\right)\text{exp}\left(\frac{-jw\mu_{I}^2}{1+2jw\sigma_{I}^2}\right)}{(1-2jw\sigma^2_{A})^{0.5}(1+2jw\sigma^2_{I})},
\end{align}	
where ($\mu_{A}$,$\sigma_{A}$) and ($\mu_{I}$, $\sigma_{I}$) are the mean and standard deviation pairs of $A_m$ and $\bar{I}_m$, respectively. Thus, by using Gil-Pelaez’s inversion formula, $P_m^{out}$ can be obtained as follows\cite{Gil-cdf}
\begin{align}
	P(Y<y)&=\frac{1}{2}-\int_{0}^{\infty}\frac{\Im\{e^{-jwy}\Psi_{Y}(w)\}}{w\pi}dw,\label{eq:Gaz}
\end{align}
	where the integration needs to be evaluated numerically with a suitable upper limit to avoid errors associated with the numerical calculations.
\end{remark}
%%%%%%%%%End of  Outage Probability%%%%%%%%%%%
%%%%%%%%%%%%%%%%%%%%%%%%%%%%%%%
%%%%%%%%%%%%%%%%%%%%%%%%%
 %%%%%%%%%%%%%%Partitioning Approach%%%%%%%%%%%%%%%%%
\section{RIS Partitioning Approach and Special Cases}\label{sec:SpltAprh}
In this section, we provide the approach we used to partition the RIS among $\text{C}_2$ users in order to maximize user fairness. Since each sub-surface of the RIS is allocated to serve a different user, each user $m$ in $\text{C}_2$ receives a huge number of interferer signals from the other sub-surfaces allocated to the other users. Therefore, a proper number of reflecting elements ($N_m$) needs to be allocated for each user in order to guarantee maximum user fairness among them. In order to achieve this goal, we formulate an optimization problem where the Jain's fairness index \cite{Jain} is the objective function to be maximized and $N_1, ..., N_m, ..., N_{M_2}$ are the decision variables to be determined, as follows:
%%%%%%%%%%%%%%%%%%% (P1) %%%%%%%%%%%%%%%%%%
\begin{align}
	\hspace{-.0cm}(\text{P1}):\;\;&\smash{\underset{N_1, ..., N_m, ...,  N_{M_2}}{\text{max}}}\;\;\;\;\;\;\frac{(\frac{1}{M_2}\sum_{m=1}^{M_2}\bar{R}_m)^2}{\frac{1}{M_2}\sum_{m=1}^{M_2}\bar{R}_m^2},\tag{20a}\\
	&\hspace{-1cm}\text{s.t.} &\hspace{-6.5cm}\sum_{m=1}^{M_2}N_m=N,\;N_m\in\{1, ..., N-(M_2-1)\}, \forall m, \tag{20b}
\end{align}
where $\bar{R}_m= \mathbb{E}[R_m]$ is the ergodic rate of $\text{U}_{m,2}$, $\mathbb{E}[\cdot]$ is the statistical expectation operator over all the random channel realizations. 

(P1) is NLIP problem with a non-convex feasible set represented by the constraint (20b) that has combinatorial growth as $N$ increases. This makes (P1) a non-convex combinatorial optimization problem which is very challenging to solve if we consider the fact that a LIP is generally a non-deterministic polynomial-time (NP) hard problem \cite{Integer-NP}. IP optimization problems are usually solved by using branch-and-bound with the relaxation of the integrality constraint on the decision variables \cite{branchbound}, which is, in this case, invalid for $N_m$, the number of RIS elements allocated to $\text{U}_{m,2}$. On the other side, considering the exhaustive search solution, the complexity lies in the explosively large size of the feasible set represented by the constraint (20b) which has ${N-1 \choose M_2-1}$ feasible points. This implies that the required searching loops to scan all the feasible points are at a complexity level of $\mathcal{O}((N-1)^{\min(M_2-1, N-M_2)})$, or $\mathcal{O}((N-1)^{M_2-1})$ if we considered the fact that $N>>M_2$. Nevertheless, in order to shrink the size of the feasible set and thus, make the exhaustive search a practical method to solve (P1), we exploit the structure of the problem, specifically the nature of the transmission over the RIS, to bound the feasible set in (20b), as follows.

First, the number of RIS elements ($N_m$) allocated to any user $m$ in $\text{C}_2$ cannot be randomly small, otherwise, the irremovable interference power from the other sub-surfaces overwhelms the amplification power from the $ m^{th}$ sub-surface. Motivated by the so-called interference temperature \cite{Itemp}, a signal-to-interference power constraint (SIPC)\cite{CRadio} needs to be considered to protect each user $\text{U}_{m,2}$ from the interference belong to the other users. Based on the SIPC threshold, there is a minimum number of RIS elements $N_{thr}$ that can be allocated for any  user $m$ in $\text{C}_2$ to protect it from the interference power associated with the remaining part of the RIS ($N-N_{thr}$). This means that the feasible set in (20b) needs to have $N_{thr}$ (rather than unity) as the minimum element in the set, which in turn shrinks the size of the feasible set to ${N-N_{thr}-1 \choose M_2-1}$ feasible points. Second, instead of considering a step size of one between any two successive elements in the feasible set of (20b), an adjustable step size $b$ is considered, which reduces the feasible set further to ${\frac{N-N_{thr}}{b}-1 \choose M_2-1}$ feasible points, where $\frac{N-N_{thr}}{b}$ needs to be an integer. Third, as in the power allocation of classical PD-NOMA, more RIS elements need to be allocated to the user with the weakest RIS-$\text{U}_{m, 2}$ channel gain. Hence, the users need to be ordered according to their distances from the RIS, where $\text{U}_{m, 2}$ is the $ m^{th}$ farthest user from the RIS and thus, the user with the $ m^{th}$ weakest RIS-$\text{U}_{m, 2}$ channel gain. By considering the previously mentioned three
bounding modifications on the feasible set in (20b), (P1) can be reformulated to obtain the following new optimization problem:
\vspace{-0.1cm}
%%%%%%%%%%%%%%%%%%% (P1.1) %%%%%%%%%%%%%%%%%%
\begin{align}
	\hspace{-.1cm}(\text{P1.1}):\;\;&\smash{\underset{N_1, ..., N_m, ...,  N_{M_2}}{\text{max}}}\;\;\;\;\;\;\frac{(\frac{1}{M_2}\sum_{m=1}^{M_2}\bar{R}_m)^2}{\frac{1}{M_2}\sum_{m=1}^{M_2}\bar{R}_m^2},\tag{21a}\\
	%%%%%%s S . T .%%%%%%%%%%%%%%%%%
	&\hspace{-1cm}\text{s.t.}
	%%%%%%%% Constraints%%%%%%%%%%%%%%%
	\hspace{0cm}\sum_{m=1}^{M_2}N_m=N,\;N_m\in\mathcal{N}, \;\forall m, \text{where}\; \mathcal{N}=\{N_{thr}, N_{thr}\nonumber\\
	&+1b, N_{thr}
	+2b,..., N-N_{thr}(M_2-1)\}, \mathcal{N}\subset\mathbb{Z}^+, \tag{21b}\\
	%%%%%%%%%%%%%
	&\hspace{0cm} N_1\geq N_2...\geq N_{M_2}.\tag{21c}
\end{align}
By considering (P1.1), it can be observed that the combinatorial growth of the feasible set of (P1) is significantly limited by the three modifications considered on the constraint (20b), which in turn effectively reduces the search space. In what follows, we formulate two new optimization problems to find the proper $N_{thr}$ and $b$ for (P1.1):
%%%%%%%%%%%%%%%%%%% (P1.1.1) %%%%%%%%%%%%%%%%%%
\begin{align}
	\hspace{-0.2cm}(\text{P1.1.1}):\;\;&\text{min}\;\;\;\;N_{thr},\tag{22a}\\
	&\hspace{-1.0cm}\text{s.t.}
	\;\;P(A_m|_{N_m=N_{thr}}\leq q\left|I_{\text{RIS}}|_{M_2=2, i\neq m}\right|^2)\leq \epsilon \tag{22b}\\
	%%%Explicit variables%%%%%%%
	%&\hspace{-5cm}P\left(\left|\sqrt{L_1^{\text{RIS}}}\sum_{n=1}^{N_{thr}} \beta_{m,m}^{(n)}\right|^2<q\left|\sqrt{L_1^{\text{RIS}}}\left[\sum_{n=1}^{N-N_{thr}}\beta_{i,m}^{(n)}e^{j(\phi_{m,m}^{(n)}+\psi^{(n)}+\theta_m)}\right]\right|^%2\right)\leq \epsilon, \tag{20b}\\
	%%%Explicit variables%%%%%%%
	&\hspace{-1.1cm}\;\;\;\;\;\;\; M_2N_{thr}\leq N. \tag{22c}
\end{align}
Here, the motivation behind the constraint (22b) can be explained by the fact that from each user $\text{U}_{m,2}$ perspective, the RIS appears to be partitioned into two parts, the first part (with $N_{thr}$ elements) that amplifies the signal of $\text{U}_{m,2}$, and the second part (with $N-N_{thr}$ elements) where the interference associated with the other users comes from. Thus, regardless of the noise and the BS-$\text{U}_{m,2}$ interference, the constraint (22b) ensures that for any user  $\text{U}_{m,2}$ with $N_{thr}$ RIS elements allocation, there is a low probability $\epsilon<<1$ that the interference power associated with the $N-N_{thr}$ part can be higher (by a factor of $q$) than the amplification power associated with the $N_{thr}$ part. The two parameters $\epsilon$ and $q$ need to be adjusted according to the size of the RIS and the number of users in $\text{C}_2$.  Furthermore, the probability given in (22b) can be obtained from \eqref{eq:OP_c1} by letting $L_m^{\text{BS}}=0$ and $\gamma_m^*=1$, as it is illustrated in Algorithm 1. The constraint (22c) ensures that for a given $N$, $N_{thr}$ needs to have a small enough value such that all the $M_2$ users can have (at least) the same allocation of $N_{thr}$ RIS elements, otherwise, $N$ need to be increased. In what follows we formulate an optimization problem to find the proper step size $b$ for (P1.1).
%%%%%%%%%%%% (P1.1.2) %%%%%%%%%%%%%%%%%%
\begin{align}
	\hspace{-.1cm}(\text{P1.1.2}):\hspace{0.5cm}\text{max}&\hspace{0.5cm}b,\tag{23a}\\
	\hspace{0.3cm}\text{s.t.}& \hspace{0.4cm}\bar{R}_m^{(N_{thr}+b)}-\bar{R}_m^{(N_{thr})}\leq \bar{r}, \tag{23b}\\
	&\hspace{0.45cm}b\leq N-M_2N_{thr},\tag{23c}
\end{align}
%%%setting eq. counter%%%%%%%%%%
\setcounter{equation}{23}
\hspace{0.17cm}where $\bar{R}_m^{(N_{thr}+b)}$ and $\bar{R}_m^{(N_{thr})}$ are obtained from \eqref{eq:R} with simple modifications, as illustrated in Algorithm 2.

After determining $N_{thr}$ in (P1.1.1), in (P1.1.2) we aim to find the maximum increment $b$  that can be added to $N_{thr}$ to get, accordingly, an ergodic rate increment upper-bounded by $\bar{r}$ b/s/Hz. Thus, $b$ corresponds to the step size of the exhaustive search that ensures a maximum of $\bar{r}$ ergodic-rate resolution. In (23b), due to the fact that the users in $\text{C}_2$ experience different SNR values and different BS-$\text{U}_{m, 2}$ interference power levels, the ergodic rate difference is calculated in the absence of the noise and the BS-$\text{U}_{m, 2}$ interference effects. In this way, $b$ determined from (P1.1.2) ensures that in the presence of noise and BS-$\text{U}_{m, 2}$ interference, the search resolution to solve (P1.1) does not exceed $\bar{r}$, for all users. Hence, $\bar{r}$ needs to be adjusted according to the RIS size and number of users in $\text{C}_2$.
\vspace{-0.2cm}
%%%%%%%%%%%%%%%%%%%%%%%%%%%%%%%%%
%%%%%%%%%%%%%%%%%%%%%%%%%%%%%%%
%%%%%%%%%%%%%%%%%%%%%%%%%%%%%%%%
 \subsection{RIS Partitioning Algorithms}\label{sec:SpltAlg}
Note that Algorithm 3 is the main algorithm to obtain the RIS sub-optimum partition $N_1, N_2,..., N_{M_2}$. Nevertheless, Algorithms 1, 2, and 3 need to be applied in order, where Algorithm 1 is used first to obtain $N_{thr}$, which is the input of Algorithm 2. In the same way, Algorithm 2 is used to obtain $b$, then, $N_{thr}$ and $b$ are used as inputs to Algorithm 3, which can  be summarized as follows. First, the set $\mathcal{S}$ is constructed according to the constraints (21b) and (21c). The elements of $\mathcal{S}$ are the possible partitions the RIS can be  partitioned into. When $b=1$ is used in (21b), $\mathcal{S}$ contains all the possible partitions and the solution of the algorithm is a globally optimal solution, otherwise, for $b>1$, the solution is a sub-optimal one. Second, for each partition $s$ in $\mathcal{S}$, the ergodic transmission rate $\bar{R}_m$ for the all $M_2$ users are calculated and stored in the set $\mathcal{R}^{(s)}$. Third, for each partition $s$, Jain's index is calculated  from $\mathcal{R}^{(s)}$ and stored in the set $\mathcal{J}$. Finally, the partition $j^*$ associated with the maximum Jain's index is chosen as the optimum partition, and the set $\tilde{\mathcal{S}}^{(j^*)}$ contains the optimum number of RIS elements $N_m^*$ needs to be allocated to each user $m$ in $\text{C}_2$.

The convergence of the three algorithms is guaranteed as all algorithms have a predetermined finite number of iterations. Specifically, Algorithms 1 and 2 have $\bar{I}_1=N/M_2$ and $\bar{I}_2=N-M_2N_{thr}$ maximum number of iterations, respectively. Likewise, Algorithm 3 has a maximum upper bound number of iterations $\bar{I}_3=2{\frac{N-N_{thr}}{b}-1 \choose M_2-1}$, which can be verified from the constraints (21b) and (21c) of (P1.1).
%%%%%%%%%%%%%%%% Alg. 1: Find N_thr %%%%%%%%%%%%%%%%%
\begin{algorithm}
	%\begin{scriptsize}
	\caption{Solves (P1.1.1) to find $N_{thr}$.}
	\begin{algorithmic}[1]
		\REQUIRE $L_m^{\text{RIS}}$, $M_2$, $N$, $q$, $\epsilon$.
		\STATE Initialize $N_{thr}=1, m=1, i=2$.
		\STATE \textbf{repeat}
		\STATE $\mu_m=\sqrt{L_m^{\text{RIS}}}N_{thr}\frac{\sqrt{\pi}}{2}, s_m^2=L_m^{\text{RIS}}N_{thr}\frac{4-\pi}{4}$,  $s_i^2=0.5L_m^{\text{RIS}}q(N-N_{thr})$.
		\STATE The probability in (22b) can be obtained from \eqref{eq:OP_c1}:\\ $P_m^{\infty} 	=\left(\frac{s_i^2}{s_m^2+s_i^2}\right)^{\frac{1}{2}}\text{exp}\left(-\frac{\mu_m^2}{2(s_m^2+s_i^2)}\right).$
		\STATE $N_{thr}=N_{thr}+1$.
		\STATE \textbf{while} $P_m^{\infty}<=\epsilon$ and $M_2N_{thr}\leq N$.
		\RETURN $N_{thr}=N_{thr}-1$.
	\end{algorithmic}
	% \end{scriptsize}
\end{algorithm}
%%%%%%%%%%%%%%%%%%%%%%%%
%%%%%%First case%%%%%%%%
%%%%%%%%%%%%%%%% Alg. 2: Find the step size b %%%%%%%%%%%%%%%%%
\begin{algorithm}
	%\begin{scriptsize}
	\caption{Solves (P1.1.2) to find the step size $b$.}
	\begin{algorithmic}[1]
		\REQUIRE $N$, $N_{thr}$, $\bar{r}$, $\beta_{m,m}^{(n)}$, $\beta_{i,m}^{(n)}$, $\bar{\Phi}_i^{n}$, $\forall n$, and for any $ m$ and $i$ such that $i\neq m$.
		\STATE Initialize $b=1$.
		\STATE $\bar{R}_m^{(N_{thr})}=\mathbb{E}\left[\log_2\left(1+\frac{A_m|_{N_m=N_{thr}}}
		{|I_{\text{RIS}}|_{N_i=N-N_{thr}}|^2}\right)\right]$. Perform the expectation over $10^4$ random channel realizations \cite{Ergodic}.
		\STATE \textbf{repeat}
		\STATE $\bar{R}_m^{(N_{thr}+b)}=\mathbb{E}\left[\log_2\left(1+\frac{A_m|_{N_m=N_{thr}+b}}
		{|I_{\text{RIS}}|_{N_i=N-N_{thr}-b}|^2}\right)\right]$.
		\STATE $\bar{r}_{diff}=\bar{R}_m^{(N_{thr}+b)}-\bar{R}_m^{(N_{thr})}$.
		\STATE $b=b+1$.
		\STATE \textbf{while} $\bar{r}_{diff}\leq\bar{r}$ and $b\leq N-M_2N_{thr}$.
		\RETURN $b=b-1$.
	\end{algorithmic}
	% \end{scriptsize}
\end{algorithm}
%%%%%%%%%%
%%%%%%%%%%%Algorithm 3: RIS Partitioning%%%%%%%%%%%%%%%%%%%
\begin{algorithm}
	%\begin{scriptsize}
	\caption{RIS partitioning algorithm to solve (P1.1) }
	\begin{algorithmic}[1]
		\REQUIRE $b$, $N_{thr}$, $N$,  $M_2$, $L^{\text{RIS}}_{m}$, $\rho$, $v_m$, $\beta_{i,m}^{(n)}$, $\phi_{i,m}^{(n)}$, $\bar{\Phi}_i^{n}$, $\forall i, n, m$. \\
		\STATE Construct the set $\mathcal{N}$ in the constraint (21b).
		\STATE According to the constraint (21b), construct the set $\mathcal{S}$ by finding all the solutions of $\sum_{m=1}^{M_2}N_m=N$, and excluding the ones that do not satisfy (21c). Thus, $\mathcal{S}=\{\tilde{\mathcal{S}}^{(1)}, ..., \tilde{\mathcal{S}}^{(s)}, ...\tilde{\mathcal{S}}^{(|\mathcal{S}|)}\}$, where $\tilde{\mathcal{S}}^{(s)}=\{N_1^{(s)}, ...,  N_m^{(s)}, ...N_M^{(s)}\}$ corresponds to the RIS partition $s$,  $\tilde{\mathcal{S}}^{(s)}\subset \mathcal{N}$, $\forall s$. \\
		\STATE  Construct the new sets $\mathcal{R}^{(1)}, ..., \mathcal{R}^{(s)}, ..., \mathcal{R}^{(|\mathcal{S}|)}$, where  $\mathcal{R}^{(s)}=\{\emptyset\}$, $\forall s$.	\\   
		\FOR{$s=1:|\mathcal{S}|$}
		\FOR{$ m=1:M_2$}
		\STATE $	R_m=\mathbb{E}\left[\log_2\left(1+\frac{A_m}{I_m+\frac{1}{\rho}}\right)\right]$, where $N_i,\;N_m\in \tilde{\mathcal{S}}^{(s)}$, $\forall i,\;m$. Perform the expectation over $10^4$ random channel realizations.
		\STATE $\mathcal{R}^{(s)}=\mathcal{R}^{(s)}\cup \{R_m\}$.
		\ENDFOR
		\ENDFOR
		\STATE Construct a new set  $\mathcal{J}=\{\emptyset\}$.\\
		\FOR{$s=1:|\mathcal{S}|$}
		\STATE $J=\frac{(\frac{1}{M_2}\sum_{m=1}^{M_2}R_m)^2}{\frac{1}{M_2}\sum_{m=1}^{M_2}R_m^2}$, where $R_m\in\mathcal{R}^{(s)}, \forall m$.
		\STATE $\mathcal{J}=\mathcal{J}\cup \{J\} $.
		\ENDFOR
		\STATE $j^*=\smash[b]{\underset{j=1:|\mathcal{J}|}{\text{arg max}}}\;\mathcal{J}^{(j)}$
		\vspace{0.5cm}
		\RETURN $\tilde{\mathcal{S}}^{(j^*)}=\{N_1^*, N_2^*, ..., N_m^*, ..., N_{M_2}^*\}$.
	\end{algorithmic}
	% \end{scriptsize}
\end{algorithm}
%%%%%%%%%%%%%%%%%%%%%%%%%%%%%%%%%%%%%%%%%%
%%%%%%%%%%%%%%%%%%%%%%%%%%%%%%%%%%%%
\hspace{0.18cm}From the convergence analysis above, the computational complexity for Algorithm $i$, $i\in\{1,2,3\}$, can be given as $\mathcal{O}(\bar{C}_i\bar{I}_i)$, where $\bar{C}_i$ is the computational complexity of the functions inside the loops  and $\bar{I}_i$ is the number of iterations for all loops, which is given above for each algorithm. By considering the required number of complex multiplications (CMs) as a metric, we obtain $N^2+N$ as the number of CMs required to construct $A_m$ and $I_m$/$I_\text{RIS}$ in Algorithms 1 and 2, where we considered the vector-matrix multiplication form given in $\eqref{eqSplit1}$. Consequently, for Algorithm 2, we obtain $\bar{I}_2(N^2+N)=(N-M_2N_{thr})(N^2+N)$ required number of CMs, hence, a complexity level of $\mathcal{O}_2(N^3)$. Likewise, for Algorithm 3, we have $0.5\bar{I}_3(N^2+N)={\frac{N-N_{thr}}{b}-1 \choose M_2-1}(N^2+N)$, hence, a complexity level of $\mathcal{O}_2(N^2{\frac{N-N_{thr}}{b}-1 \choose M_2-1})$, where we considered only the first loop that requires CMs. For Algorithm 1, which has no CMs, the complexity level associated with $\bar{C}_1$ depends on the type of the used algorithm.

\begin{remark}
Note that the partitioning process  (Algorithms 1, 2, and 3) needs to be updated only when the RIS size $N$, transmit power $P$, number of users $M_2$ in $\text{C}_2$, or the distances of users from the RIS/BS changes. This significantly reduces the computational complexity cost associated with performing these algorithms as the mentioned parameters slowly change with time. Furthermore, a minimum data rate requirement for all or individual users can be straightforwardly included in the algorithms. Particularly, such a constraint can be included in Algorithm 1 to replace the constraint (22b).
\end{remark}
\vspace{-0.4cm}
\subsection{Special Cases with Different Number of Users in Clusters}\label{sec:Spec}
Here, in addition to the general case introduced in Section \ref{sec:Main}, we consider the system performance analysis for particular cases in order to shed some light on the performance of the proposed scheme under different settings, as follows.\\
\indent i) \textit{All users are located in} $\text{C}_2$: In this scenario $M_1=0$, and the same transmission mechanism described in Section \ref{sec:Main} is used here with the following single modification. Since there are no users in $\text{C}_1$, the BS is assumed to transmit $x=\sqrt{P}\bar{x}$, where $\bar{x}$ is a predetermined symbol. Thus, assuming $v_{m}$ is known at the receiver side, the signal $v_m\bar{x}$ received by each user $m$ in $\text{C}_2$ over the BS-$\text{U}_{m,2}$ link can be removed readily. Furthermore, over the RIS-$\text{U}_{m,2}$ link, each sub-surface $i=m$ of the RIS remodulates $\bar{x}$ to reflect the PSK symbol to be transmitted to $\text{U}_{m,2}$, as described in Section \ref{sec:Main}. Without loss of generality, an unmodulated carrier signal can be sent from the BS and, in this case, the BS can be compensated by a single RF signal generator (SG), which simplifies the transmitter architecture. Furthermore, $R_m$ and $P_m^{\text{out}}$ can be obtained from \eqref{eq:R} and \eqref{eq:OP1}, respectively, by letting $v_m=0$ ($L_m^{\text{BS}}=0$). Likewise, for the asymptotic behaviour at high SNR, $P_{\text{out}}^{\infty}$ can be directly obtained from \eqref{eq:OP_c2} and \eqref{eq:OP_c22} by considering the uniform partitioning of the RIS elements to guarantee maximum user fairness.\\ 
%%%%%%%Second Case%%%%%%%%%%%%%%%%%%%%%
\indent ii) \textit{Each cluster has a single user}: In this scenario, $\text{U}_{1,1}$ and $\text{U}_{1,2}$ are the only users in $\text{C}_1$ and $\text{C}_2$, respectively. Therefore, the BS transmits $x=\sqrt{P}x_1$, and $\text{U}_{1,1}$ experience the same performance in a SISO system. On the other side, the RIS (as a single unit) remodulates $x$ and reflects the PSK symbol to be transmitted to $\text{U}_{1,2}$, as described in Section \ref{sec:Main}. Since the whole surface is allocated for $\text{U}_{1,2}$, $R_1$, $P^{\text{out}}_m$, and  $P_m^{\infty}$ can be obtained from \eqref{eq:R}, \eqref{eq:OP1}, and \eqref{eq:OP_c1}, respectively, by omitting the sub-surfaces interference term $I_{\text{RIS}}=0$ and letting $N_1=N$. 
\vspace{-0.1cm}
%%%%%%%%%%%%%%%%%%%%%%%%
%%%%%%%%%%%Simulations%%%%%%%%%%%%%%%%%%%%%%
\section{Simulation Results}\label{sec:Simu}
In this section, we provide computer simulation results for the proposed scheme against the following four benchmark schemes: the time division multiple access (TDMA) as an OMA scheme, where the BS serves, with full power $P$, each user in its time slot with/without the RIS. PD-NOMA scheme, where the SC message is constructed as in \eqref{eq:x}. Finally, the RIS-assisted SISO NOMA scheme proposed in \cite{Benchmark}, \cite{Bench2}\footnote{This work does not consider the spatial correlation for the RIS channels in the obtained power allocation solution. Nevertheless, spatially correlated channels are used in our simulations of this benchmark scheme.}. For TDMA and classical NOMA schemes, in order to achieve maximum user fairness, we optimize the time and power allocation, respectively, at the transmitter side using exhaustive search. In order to guarantee a fair comparison, we consider the maximum fairness between users as the common threshold for all schemes, while the comparison lies in the ergodic transmission rates and outage probability performance. In what follows, we describe the path loss models and other simulation parameters.\\
\indent For the BS-$\text{U}_{m,c}$ link, where $c$ denotes the cluster number, we obtain the path loss as $	(L^{\text{\text{BS}}}_{m,c})^{-1}=C_0\left(d_{m,c}/D_0\right)^{\alpha} \label{eq:pathL1}$ \cite{RIS-Reflct}, where $C_0=-30$ dB is the path gain at the reference distance $D_0=1\;$ m, $d_{m,c}$ is the BS-$\text{U}_{m,c}$ distance, and $\alpha=-3.5$ is the path gain exponent. For the BS-RIS-$\text{U}_{m,2}$ link, we consider that the RIS is located in the far-field with respect to the BS by ensuring that $r_s=\lceil \frac{N\lambda}{2}\rceil$ \cite{Expermintal-RIS}, thus, we obtain the path loss (with maximum gain RIS elements) as $(L^{\text{RIS}}_m)^{-1}=\lambda^4/(256\pi^2r_s^2r_{m}^2)\label{eq:pathL2}$ \cite{Ellingson}, where $r_s$ and $r_{m}$ are the BS-RIS (the center point) and the RIS (the center point)-$\text{U}_{m,2}$ distances, all in meters, respectively. Due to the small spacing distances between the RIS elements compared to $r_s$ and $r_m$, we consider $r_s$ and $r_m$ for the path loss calculations, consequently, the path losses for all the RIS elements are considered the same. In Table I, we give the users' distances from the BS and RIS, for different deployment scenarios, which are the used distances in our simulations, unless stated otherwise. The entries of the spatial correlation matrix are given as $[\mathbf{R}_{\text{RIS}}]_{n,\tilde{n}}=\text{sinc}\left(2\norm{\mathbf{u}_n-\mathbf{u}_{\tilde{n}}}/\lambda\right),\; n,\tilde{n}=1,...,N$ \cite{Corr},
where $\lambda$ is calculated for a $1.8$ Ghz operating frequency, $\text{sinc}(a) = \text{sin}(\pi a)/(\pi a)$ is the sinc function, $\mathbf{u}_n=[0,\;i(n)d_l,\tilde{i}(n)d_w\;]^T$, where $i(n)$, $\tilde{i}(n)$, $d_l$, and $d_w$ are the horizontal index, vertical index, length, and width of the element $n$, respectively, 
$i(n)=\text{mod}(n-1,N_H)$, $\tilde{i}(n)=\left\lfloor(n-1)/N_H\right\rfloor$. Furthermore, the noise power is assumed to be fixed and identical for all users in $\text{C}_1$ and $\text{C}_2$, $\sigma^2=-90$ dBm, $\forall m$, and the simulations are performed with $10^4$ random channel realizations.\\ 
 \indent Fig. \ref{fig:AlgPerf} illustrates the performance of the sub-optimal solution to (P1), which is represented by Algorithms 1, 2, and 3 to solve (P1.1) (the reformulated version of (P1)). 
 %%%%%%%SuOpt-Alg, performance %%%%%%%%%%
 \begin{figure}[t!]
 	\begin{center}
 		\includegraphics[width=60mm, height=40mm]{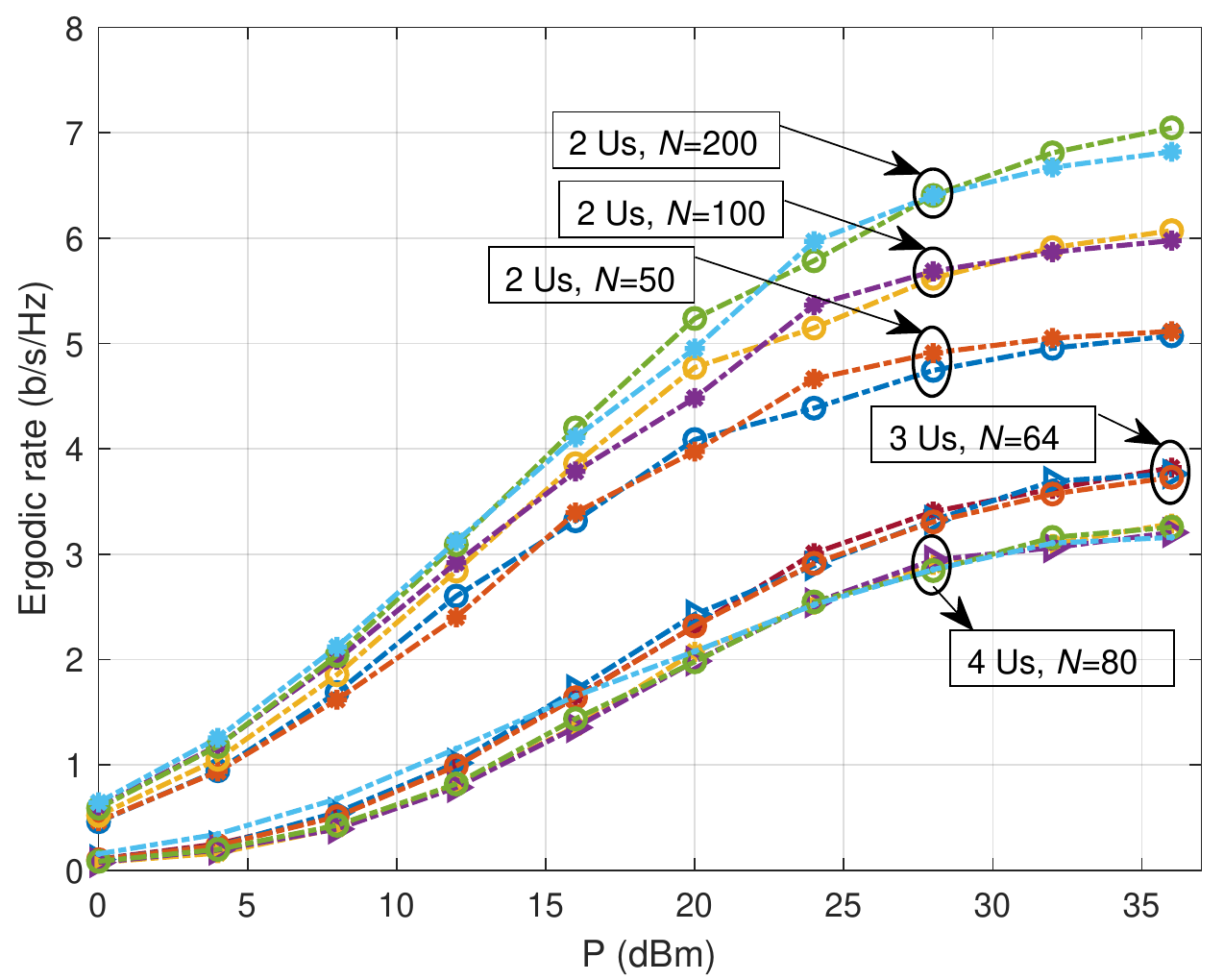}
 		\caption{The  performance of the RIS partitioning algorithms with different $N$ and $M_2$ values.}		\label{fig:AlgPerf}
 		\vspace{-0.5cm}
 	\end{center}
 \end{figure}
 \begin{table}[t]
 	%\vspace{-1.5cm}
 	\caption{Users' distances for different deployment scenarios.}
 	\label{tab:table-name}
 	\vspace{-0.3cm}
 	\begin{center}		
 		\begin{tabular}{ | m{7.7em} | m{3.8em}| m{3.9em}| m{3.9em}| m{3.9em}| } 
 			\hline
 			\textit{Users' deployment} & $d_{1,1}$, $r_{1,1}$  & $d_{2,1}$, $r_{2,1}$ & $d_{1,2}$, $r_{1,2}$ & $d_{2,2}$, $r_{2,2}$\\ 
 			\hline
 			$M_1=0$, $M_2=2$ & - & - & $150$, $146$ & $100$, $104$\\ 
 			\hline
 			$M_1=1$, $M_2=1$ & $150$, $146$ & - & $100$, $104$ & -\\ 
 			\hline
 			$M_1=2$, $M_2=2$ & $150$, $154$ & $100$, $104$ & $250$,  $254$ & $200$, $204$\\ 
 			\hline
 		\end{tabular}
 	\end{center}
 	\vspace{-0.5cm}
 \end{table}
 %%%%%%%%%%%%%%%%%%%%%%%%%%%%%%%%%%%%%%
 %%%%%%%%%%%%%%%%%%%%%%%
 \begin{table}[t!]
 	\caption{RIS partitioning algoritms: input parameters and outcomes.}	\label{tab:states}
 	\vspace{-0.3cm}
 	\begin{center}
 		\begin{tabular}{| m{2.5em} | m{5.5em}| m{3em} | m{2.47em} | }
 			\hline
 			$N, M_2$ & $N_{thr}$, $q$, $\epsilon$ & $b$, $\bar{r}$ & $|\mathcal{S}|/|\tilde{\mathcal{S}}|$ \\ 
 			\hline
 			$50, 2$ & $15$, $1.5$, $0.1$ & $2$, $0.3$ & $5/49$\\
 			\hline  
 			$64, 3$ & $14$, $1$, $0.1$ & $1$, $0.1$ &$40/210$ \\ 
 			\hline
 			$80, 4$ & $15$, $1$, $0.1$ & $1$, $0.1$ &$64/969$ \\ 
 			\hline
 			$100, 2$ & $21$, $1.5$, $0.1$ & $5$, $0.5$ & $5/99$ \\  
 			\hline
 			$200, 2$ & $32$, $1.5$, $0.05$ & $10$, $0.7$ &$6/199$ \\  
 			\hline
 		\end{tabular}
 	\end{center}
 	\vspace{-0.6cm}
 \end{table}\hspace{-0.15cm}
 Here, for each given $P$ value, we consider multiple users in $\text{C}_2$ where the RIS needs to be partitioned among them to maximize the user fairness in terms of the ergodic-rate. It can be seen that the proposed algorithms achieve a result close to the perfect user fairness between the users with different $N$ values. Furthermore, Table II shows that Algorithms 1 and 2 reduce the search space for Algorithm 3 significantly, where $|\mathcal{S}|/|\tilde{\mathcal{S}}|$ denote the number of possible partitions provided by the sub-optimal/optimal solution of Algorithm 3. This makes the exhaustive search a practical solution to solve (P1) and find the proper RIS partition.\\
 	%%%%%%%%%%%%%%%%%% Case1: Proof, OP%%%%%%%%%%%%
 \begin{figure}[t!]
 	\begin{center}
 		\subfloat[]{\includegraphics[width=45mm,height=45mm]{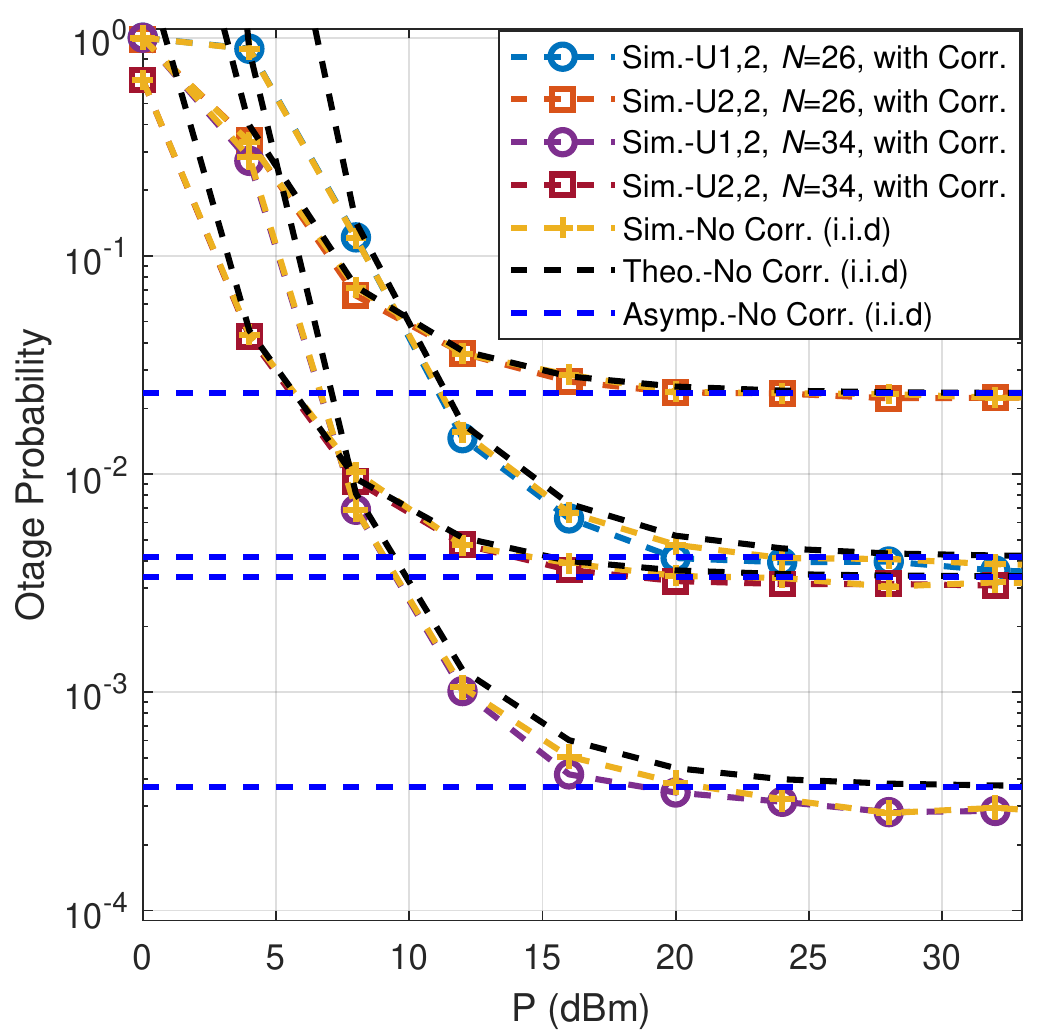}}
 		%	\subfloat[]{\includegraphics[width=45mm, height=50mm]{Figures/Theory_Proof_OP_Case1_R2.pdf}}
 		%	\hfil
 		%	\subfloat[]{\includegraphics[width=45mm,height=50mm]{Figures/Theory_Proof_OP_Case2_R2.pdf}}
 		\subfloat[]{\includegraphics[width=45mm,height=45mm]{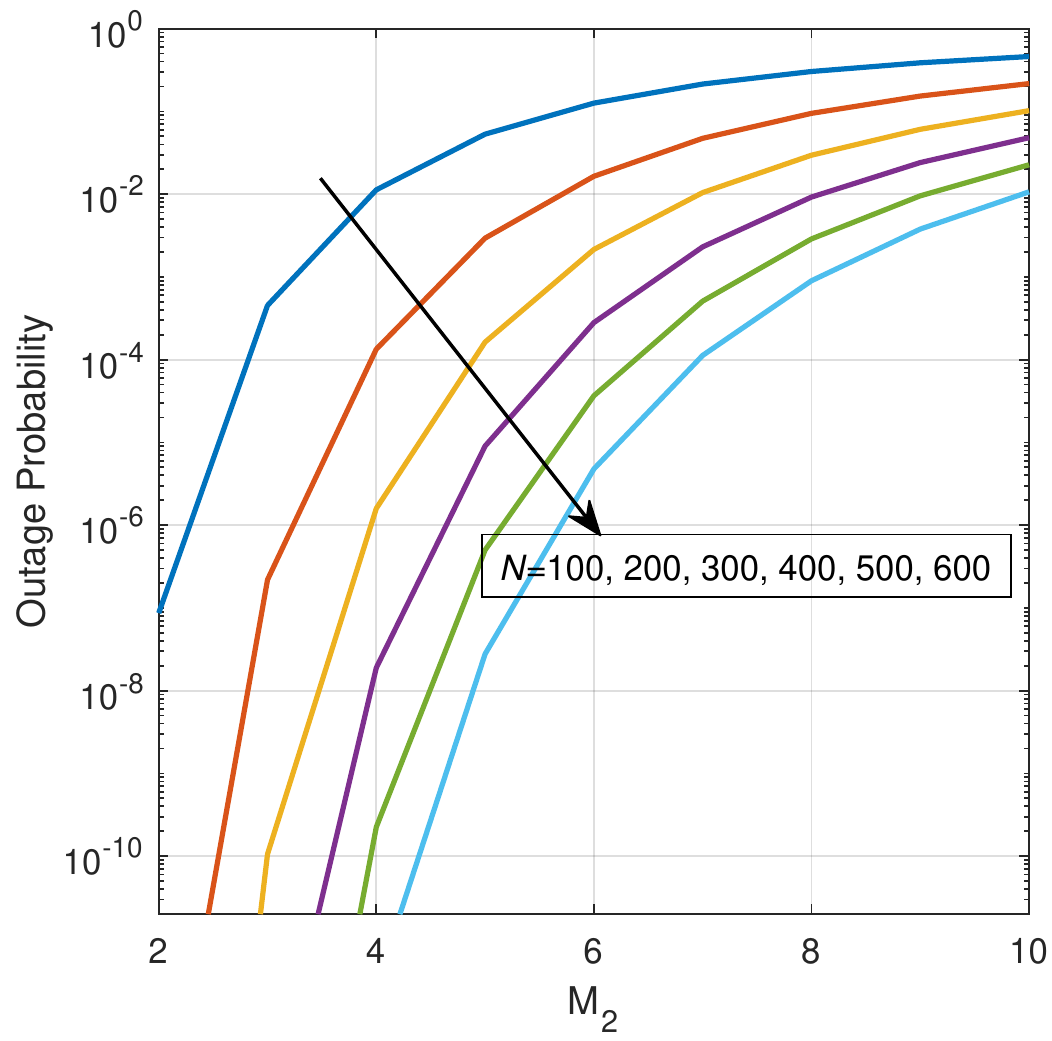}}
 		\caption{Outage probability performance with different $N$ values, for (a) the general case, $M_2=2$, $M_1\geq1$, and a targeted transmission rate of $0.75$ bpcu for all users, (b) different numbers of users in $\text{C}_2$ with uniform partitioning, obtained from \eqref{eq:OP_c22}.}
 		\label{fig:OP_theo}
 	\end{center}
 	\vspace{-0.90cm}
 \end{figure}
\indent Fig. \ref{fig:OP_theo}(a) shows the outage probability performance for the general case, with different $N$ values. Note that, as the performance improves with the increase of $N$, the saturation happens at the high SNR region due to the interference coming from BS and sub-surfaces. It can be also seen that the theoretical and simulation curves are in perfect agreement with each other. Furthermore, it can be noted that the i.i.d.  approximation of the RIS-$\text{U}_{m,2}$ channels is accurate, where the curves with/without spatial correlation are perfectly matched. Fig. \ref{fig:OP_theo}(b) shows the outage probability performance of users in $\text{C}_2$ versus the number of users $M_2$, where the RIS is partitioned uniformly between all users as stated in Corollary 2. It is worth noting that, although the ratio of the amplification to the interferer signals is the same, increasing the RIS size enhances the performance of all users due to the nature of the incoherent interference.
In what follows we compare the proposed scheme with different benchmark schemes by considering the two special cases introduced in Section \ref{sec:Spec} and the general case introduced in Section \ref{sec:Main}. 

\indent In Figs. \ref{fig:Us2_Rate_1} (a) and (b), we consider the first case with $N=40$, $M_1=0$, and $M_2=2$. As it is shown in Fig. \ref{fig:Us2_Rate_1}(a), all the schemes achieve almost the same user fairness, nevertheless, the proposed scheme shines out with an improvement in the required $P$ of $8$ dB  compared to the RIS-NOMA scheme and $14$ dB compared to TDMA-OMA and PD-NOMA schemes. However, the performance gain achieved by the proposed scheme is bounded by a saturation point due to the mutual sub-surfaces interference.
%%%%%%%FIG: Sum-Rate-All, Case 1 %%%%%%%%%%
%%%%%%%%%%%%%%%%%%%%%
\begin{figure}
	\centering     
	%%%%%%%%%%%%%Case1:: R A T E S%%%%%%%%%%%%%%%%%%%
	\subfloat[]{\includegraphics[width=45mm, height=45mm]{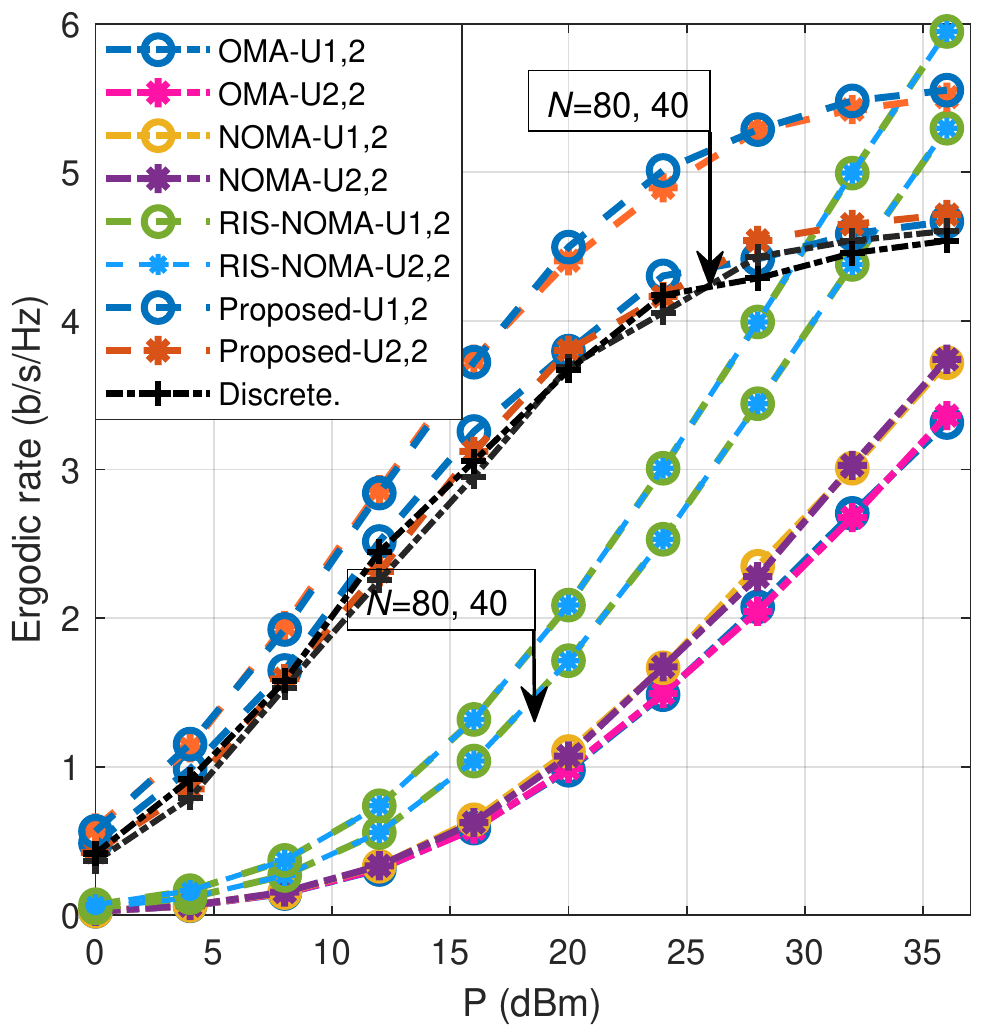}}
	%%%%%%%Case1::Sum Rate%%%%%%%%%%%%%%%%%%%%%
	% 	\subfloat[]{\label{fig1:b}\includegraphics[width=45mm, height=50mm]{Figures/Case1_SumRate.pdf}}
	%	\hfill
	%%%%%%%Case1::OP%%%%%%%%%%
	\subfloat[]{\label{fig1:b}\includegraphics[width=45mm, height=45mm]{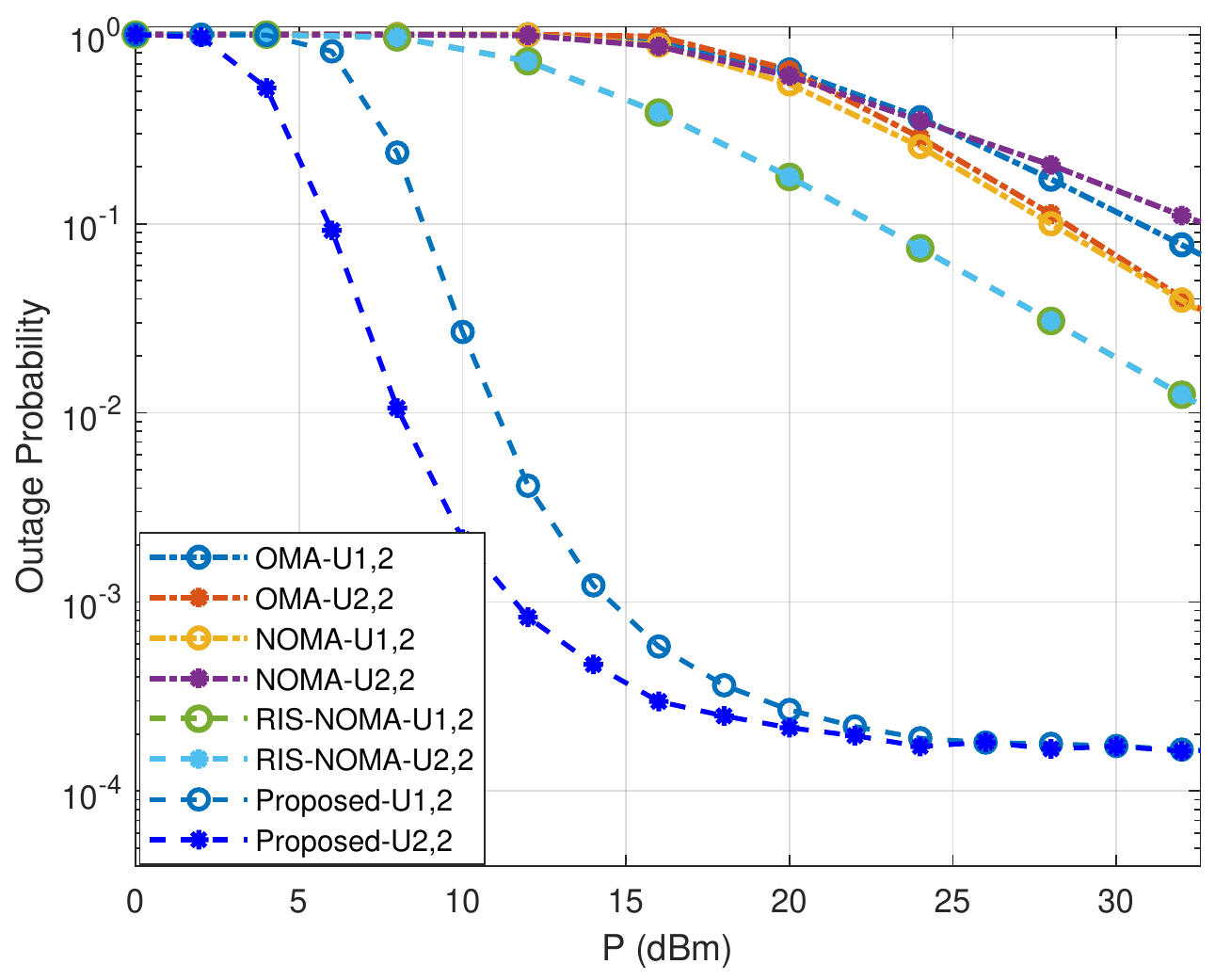}}
	%%%%%%%%%%%%%%%%%%%%%%%%%%%%%%%%%%%%%%%%%
	\caption{The comparison of the proposed and benchmark schemes with $M_1=0$ and $M_2=2$, in terms of (a) ergodic rate with $N=$40, (b) the outage probability, with a targeted transmission rate of $1.2$ bpcu for all users, and $N=40$.}
	\label{fig:Us2_Rate_1}
	\vspace{-0.7cm}
\end{figure}
%%%%%%%$$$$$$$$$%%%%%%%%%%%%%%
%%%%%%%$$$$$$$$$%%%%%%%%%%%%%%
%%%%%%%FIG: Sum-Rate-All, Case 2 %%%%%%%%%%
\begin{figure}[t!]
	\centering     
	%%%%%%%%%%%%%Case2:: R A T E S%%%%%%%%%%%%%%%%%%%
	\subfloat[]{\includegraphics[width=45mm, height=40mm]{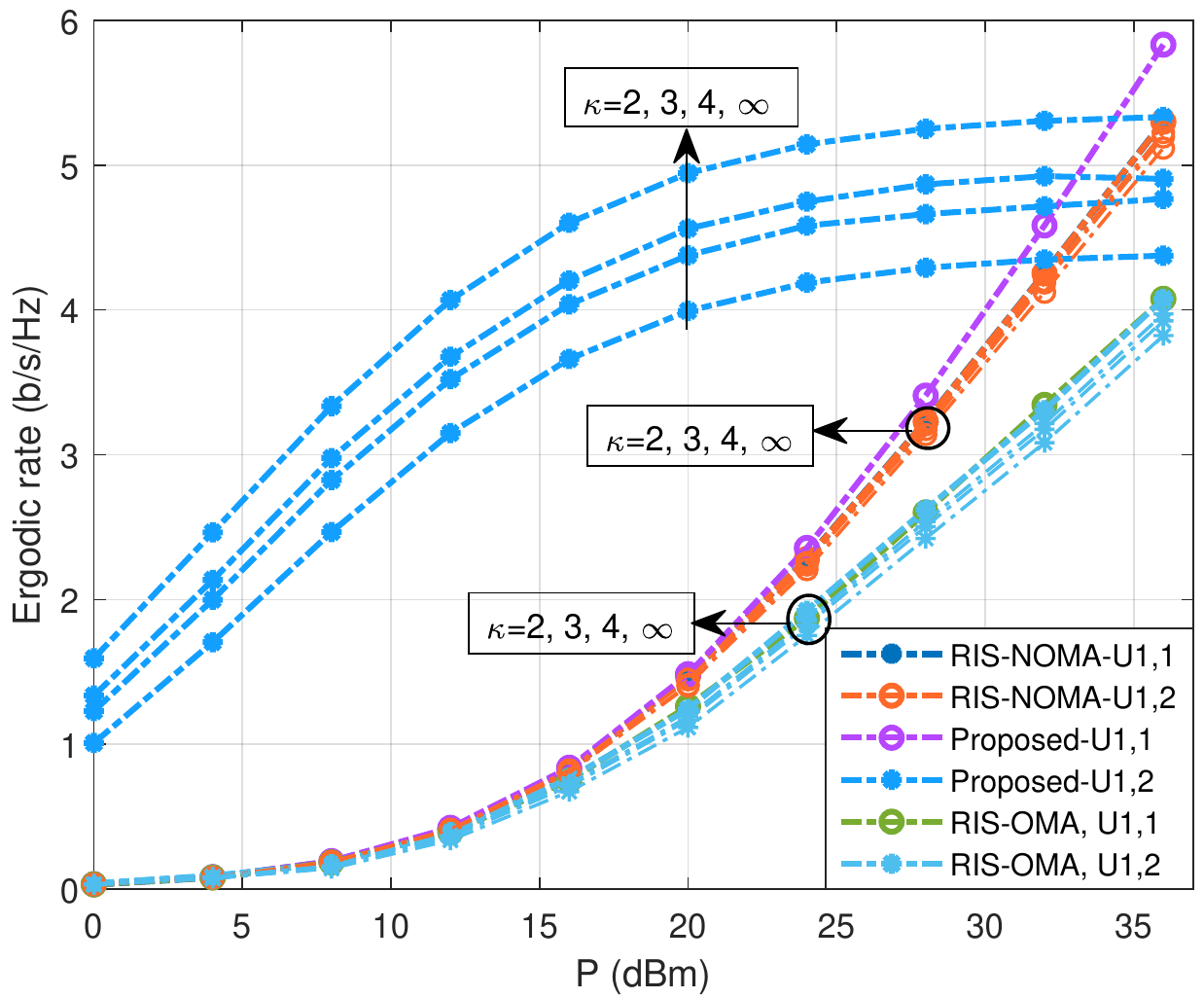}}
	%%%%%%%Case2::Sum Rate%%%%%%%%%%%%%%%%%%%%%
	%	\subfloat[]{\label{fig2:b}\includegraphics[width=55mm, height=40mm]{Figures/Case2_SumRate.pdf}}
	%%%%%%%Case2::OP%%%%%%%%%%
	\subfloat[]{\includegraphics[width=45mm, height=40mm]{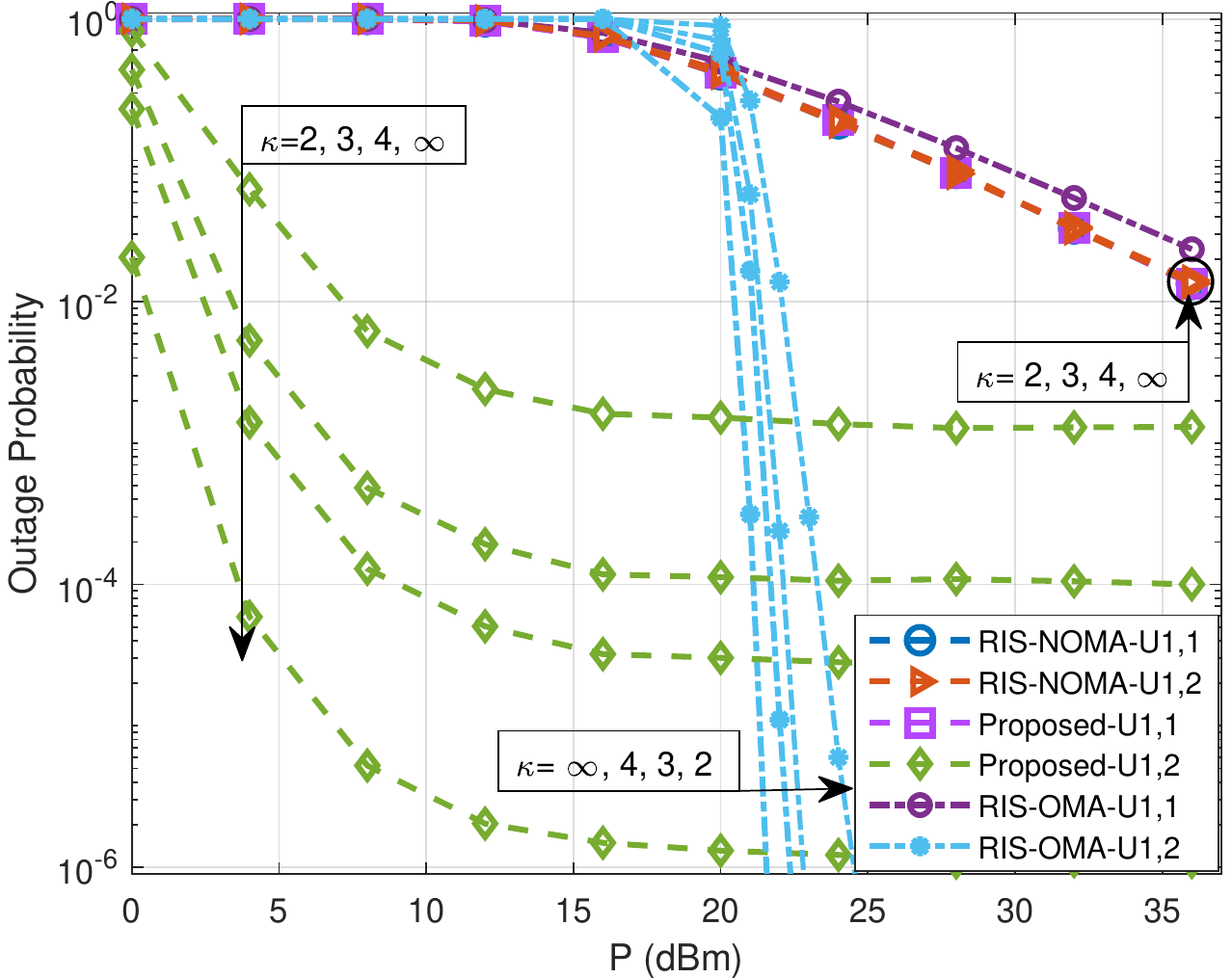}}
	%%%%%%%%%%%%%%%%%%%%%%%%%%%%%%%%%%%%%%%%%
	\caption{The comparison of the proposed scheme with the RIS-OMA (TDMA) and RIS-NOMA benchmark schemes, with $M_1=M_2=1$ and $N=40$, in terms of (a) ergodic rate, (b) outage probability, with a targeted transmission rate of $1.2$ bpcu for all users.}
	\label{fig:Us2_Rate_2}
	\vspace{-0.7cm}
\end{figure}It can be also noted that the discrete phase shift adjustment with only $8$ phase shift levels achieves almost the same performance as in the continuous case, where, assuming uniform quantization for the interval $ [0,2\pi)$, $3$ bits are used to select the phase shift from $Z=2^3$ levels in the finite set $\mathcal{F}=\{0, \Delta \Phi, ..., \Delta \Phi(Z-1)\}$, where $\Delta \Phi=\frac{2\pi}{Z}$ \cite{BeamformingPhs}. Fig \ref{fig:Us2_Rate_1} (b) shows the outage probability comparison with a targeted transmission rate of $1.2$ bpcu for all users. It can be seen that the proposed scheme outperforms the benchmark schemes with a around $20$ dB in the required $P$ for both users before the saturation region.\\
\indent In Figs. \ref{fig:Us2_Rate_2}(a) and \ref{fig:Us2_Rate_2}(b), we consider the second case with $N=40$ and $M_1=M_2=1$, as follows. Fig. \ref{fig:Us2_Rate_2}(a) shows that $\text{U}_{1,1}$ in the proposed scheme and $\text{U}_{1,1}$ and $\text{U}_{1,2}$ in the RIS-NOMA scheme, achieve almost the same ergodic rate performance, with a $2$ dB improvement in the required $P$ for the proposed scheme at the high SNR region. On the other hand, $\text{U}_{1,2}$ in the proposed scheme achieves a $12$-$24$ dB improvement compared to the other users when there is no phase estimation errors, which can be explained by the asymptotic squared power gain of $\mathcal{O}(N^2)$ the RIS provides to $\text{U}_{1,2}$ \cite{RIS-Reflct}. However, there is a saturation point for the performance of $\text{U}_{1,2}$ due to the BS-$\text{U}_{1,2}$ interference link. Considering the user fairness, contrary to might be concluded for the first glance from Fig. \ref{fig:Us2_Rate_2}(a), the proposed scheme achieves the maximum user fairness, which can be explained as follows. Since the communication link over the RIS is blocked for $\text{U}_{1,1}$ in both schemes, the most efficient option, in terms of the sum-rate performance and user fairness, is to fully allocate the RIS to serve $\text{U}_{1,2}$ and the BS to serve $\text{U}_{1,1}$ in both schemes, which is what the proposed scheme basically does. On the other side, the use of SC in the
%%%%%%%FIG: Sum-Rate-All, General Case 4 us %%%%%%%%%%
\begin{figure*}[t!]
	\centering  
	%%%%%%%%%% 4Us case:: R A T E S%%%%%%%%%%%%
	\subfloat[]{\includegraphics[width=55mm, height=45mm]{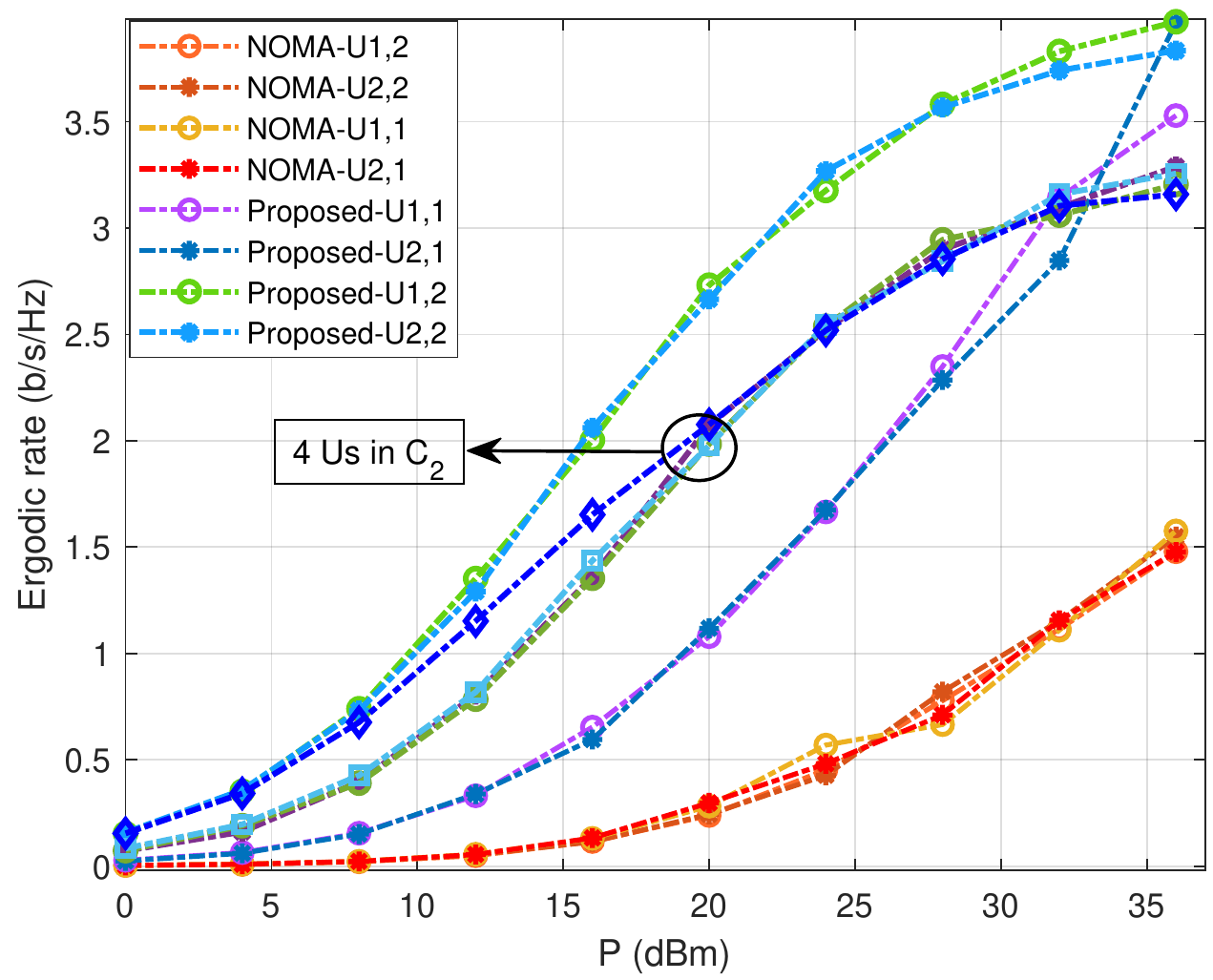}}
	%%%%%%%%% 4Us case:: S U M-R A T E S%%%%%%%%%%%%
	\subfloat[]{\includegraphics[width=55mm, height=45mm]{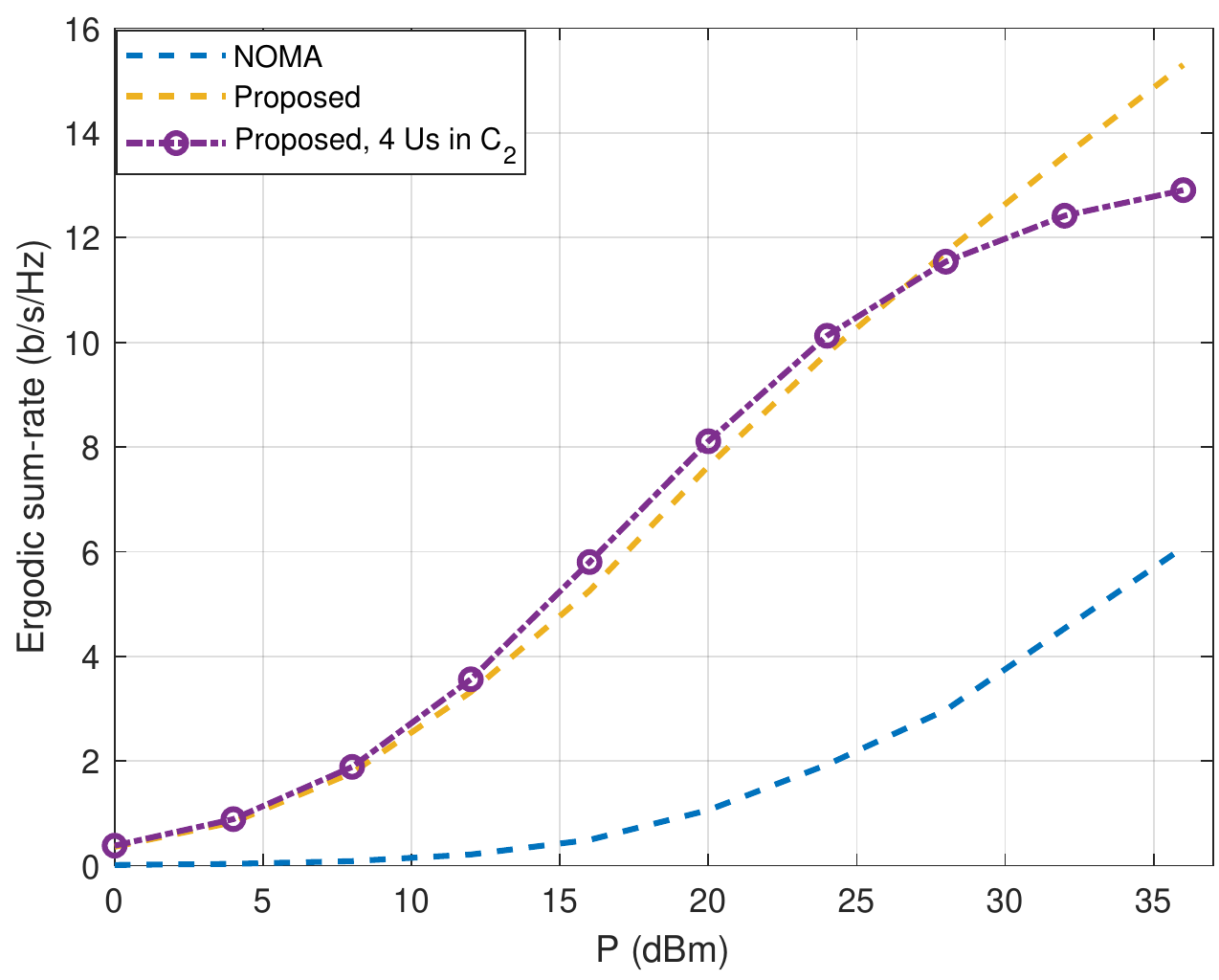}}
	%%%%%%%%%%%%% 4Us case:: OP%%%%%%%%%%%%%%%%%
	\subfloat[]{\includegraphics[width=55mm, height=45mm]{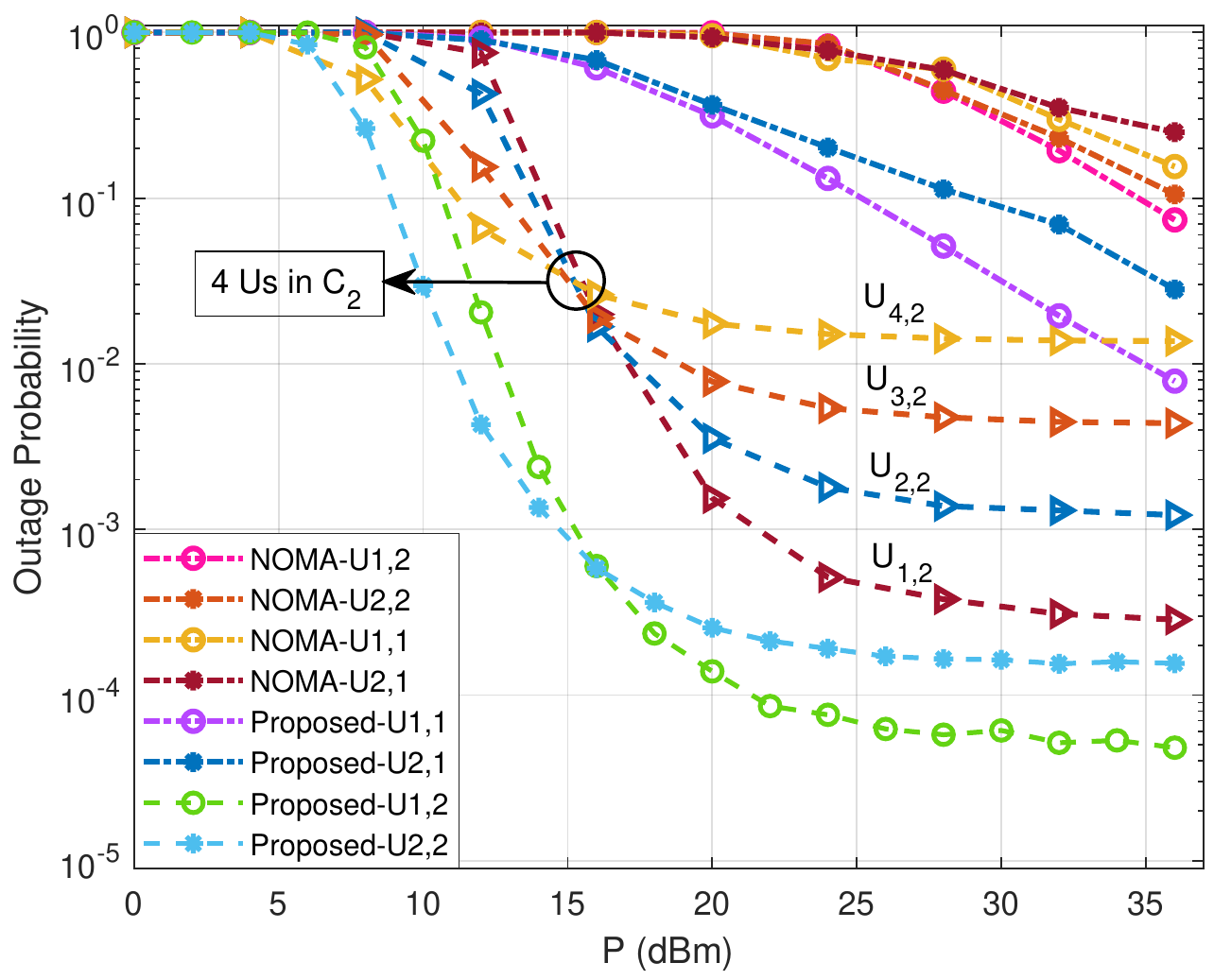}}
	%%%%%%%%%%%%%%%%%%%%%%%%%%%%%%%%%%%%%%%%%
	\caption{The comparison of the proposed and the NOMA benchmark schemes with $M_1=M_2=2$ and  $M_1=0,M_2=4$ with $N=40, 80$, respectively, in terms of (a) ergodic rate, (b) ergodic sum-rate, (c) outage probability, with a targeted transmission rate of $0.75$ bpcu for all users.}
	\label{fig:Us4_Rates}
			\vspace{-0.0cm}
\end{figure*}
 benchmark scheme makes the RIS amplifies the interference from $\text{U}_{1,1}$ to $\text{U}_{1,2}$ without any benefit for $\text{U}_{1,1}$, which is unfair to $\text{U}_{1,2}$. In Fig. \ref{fig:Us2_Rate_2}(b), $\text{U}_{1,2}$ in the proposed scheme outperforms the other users in both schemes in terms of the outage probability with a remarkable improvement of $20$-$35$ dB in the required $P$, while the other users achieve the same performance, which agrees with the ergodic rate results shown in Fig. \ref{fig:Us2_Rate_2}(a). For the phase estimation errors, we consider the von Mises distribution, where $\kappa$ is the concentration parameter. It can be seen from Figs. 5(a) and 5(b) that the benchmark schemes are considerably less sensitive to the phase estimation errors compared to the proposed scheme, which can be explained as follows. The benchmark schemes, unlike the proposed one, have a direct information link which, with this small RIS size, dominates the RIS link and therefore, the performance loss due to the phase estimation errors does not appear clearly. With a larger RIS size that makes the BS-RIS-$\text{U}_{m,2}$ reflection link dominate, the benchmark schemes are also expected to experience a similar performance behavior, with the phase estimation errors.\\
 \indent Finally, in Figs. \ref{fig:Us4_Rates}(a)-(c), we consider the general case of four users for the proposed and NOMA schemes. Fig. \ref{fig:Us4_Rates}(a) shows the ergodic rate performance for all users, where almost maximum user fairness is achieved by the NOMA scheme. On the other side, the proposed scheme provides a close to the maximum fairness performance when all users are located in $\text{C}_2$, $M_1=0,M_2=4$, with $N=80$. A better user fairness performance is noted for  $M_1=2,M_2=2,$ and $N=40$ between the two users in each cluster, which corresponds to the maximum fairness between all the four users due to the same reasons explained in the discussion of Fig. \ref{fig:Us2_Rate_2}(a). Furthermore, when $M_1=M_2=2, N=40$, the performance of the proposed scheme outperforms the one of NOMA scheme for each user significantly, with more than $12$ dB and $22$ dB improvement for the two users in $\text{C}_1$ and the two users in $\text{C}_2$, respectively. On the other side, when $M_1=0, M_2=4, N=80$, $18$ dB gain is observed compared to NOMA scheme, with a $4$ dB less gain compared to the previous case due to the mutual sub-surfaces interference. Overall, for both users' deployment scenarios, an improvement of $18$ dB in the ergodic sum-rate is achieved by the proposed scheme compared to the NOMA scheme, as shown in Fig. \ref{fig:Us4_Rates}(b). In Fig. \ref{fig:Us4_Rates}(c), when $M_1=M_2=2, N=40$, the proposed scheme outperforms NOMA scheme in the outage probability performance with around $22$ dB in the required $P$ for the first and second users in $\text{C}_2$, and around $5$ dB and $10$ dB for the first and second users in $\text{C}_1$, respectively. For the case where $M_1=0, M_2=4, N=80$, the proposed scheme outperforms the NOMA scheme with $22$ dB until the saturation region.
% \vspace{-0.5cm}
\section{Conclusion}\label{sec:Concl}
 %\vspace{-0.5cm}
In this paper, we have introduced a novel downlink NOMA solution with RIS partitioning in order to mitigate the mutual interference between users in a local beyond 5G network. In the proposed system, the ergodic rates and outage probabilities of all users are enhanced and the user fairness is maximized by the fair and efficient distribution of the PR among users. The potential of the proposed PR distribution is perfectly illustrated in Fig. \ref{fig:Us2_Rate_2}, where the BS and RIS are used in a very efficient way compared to the classical use of the SC technique. Furthermore, we have proposed three efficient searching algorithms to, sequentially, obtain a sub-optimal solution for the RIS partitioning optimization problem, with insignificant performance degradation. By considering different users' deployment scenarios, it was shown that the proposed system provides remarkable performance gain in all of the considered different environment settings. We have derived the exact and asymptotic outage probability expressions for the proposed system in all the cases including the effect of the number of users in $\text{C}_2$ and the RIS size. The computer simulations show that the proposed system outperforms the OMA, RIS-OMA, NOMA, and RIS-NOMA benchmark systems significantly, in terms of outage probability, ergodic rates of all users, and user fairness. Finally, removing the sub-surfaces and/or the BS interference for $\text{C}_2$ users appears as a future research direction that is worth investigating.
\vspace{-0cm}
\appendices
\section{Proof of Proposition 1}
From \eqref{eq:OP-SNR}, we obtain
%%%%%%%%%%%%%%%%%%%%%%%%%%%%%%%%%%%%%
\begin{align}
	P_\text{out}&=\left(\frac{A_m}{I_m+\frac{1}{\rho}}<2^{\gamma^*_m}-1\right)\nonumber \\
	&=P\left(Y<y\right)\label{eq:cdf2}
\end{align}
where $Y=A_m-\bar{I}_m$, $y=\frac{2^{\gamma_m^*}-1}{\rho}$, and $\bar{I}_m$ is given by
\begin{align}
	\bar{I}_m=(2^{\gamma_m^*}-1)I_m=|\sqrt{(2^{\gamma_m^*}-1)}(I_{\text{\text{RIS}}}+v_m)|^2\label{eq:I_m_bar}.
\end{align}
%%%%%%%%%%%%%%%%%%
From \eqref{eq:cdf2}, we note that the outage probability is equivalent to the cumulative distribution function (CDF) of $Y$. In order to find the CDF of $Y$, we first find the distribution of the RVs $A_m$ and $\bar{I}_m$, as follows. By considering \eqref{eq:A1}, we note that $\beta_{m,m}^{(n)}$ is a Rayleigh distributed RV with a mean $\text{E}[\beta_{m,m}^{(n)}]=\sqrt{\pi}/2$, and a variance $\text{VAR}[\beta_{m,m}^{(n)}]=(4-\pi)/4$.
According to the central limit theorem (CLT), for $N_m>>1$, the term inside the squared parenthesis is a Gaussian RV,  $\sim\mathcal{N}(\sqrt{L^{\text{RIS}}_{m}}N_m\frac{\sqrt{\pi}}{2},\allowbreak L^{\text{RIS}}_{m}N_m\frac{4-\pi}{4})$. Thus, $A_m$ is a non-central chi-square ($\chi^2$) RV with one degree of freedom. Likewise, by considering \eqref{eq:I_RIS}, for $N-N_m>>1$, $I_{\text{\text{RIS}}}$ is a Gaussian RV, $I_{\text{\text{RIS}}}\sim\mathcal{C}\mathcal{N}(0,L^{\text{RIS}}_{m}(N-N_m))$. Consequently, the constant-scaled sum of the two independent Gaussian RVs in \eqref{eq:I_m_bar}, $\sqrt{(2^{\gamma_m^*}-1)}(I_{\text{\text{RIS}}}+v_m)$, is a Gaussian RV, $\sim\mathcal{C}\mathcal{N}(0,(2^{\gamma_m^*}-1)(L^{\text{RIS}}_{m}(N-N_m)+L_m^{\text{BS}})$. Thus, $\bar{I}_m$ is a central $\chi^2$ RV with two degrees of freedom, and consequently, $Y$ is the difference of a non-central and central independent $\chi^2$ RVs, where its CDF is given in \eqref{eq:OP1} \cite{Involv_GaussRVs}. This completes the proof of Proposition 1.\hspace{8.3cm}\qedsymbol   
% or
%\appendices
%\section{}
%\ifCLASSOPTIONcaptionsoff
 % \newpage
%\fi
\bibliographystyle{IEEEtran}
\bibliography{IEEEabrv,Bibliography}
\vspace{1.7cm}
\begin{IEEEbiography}[{\includegraphics[width=1in,height=1.25in,clip,keepaspectratio]{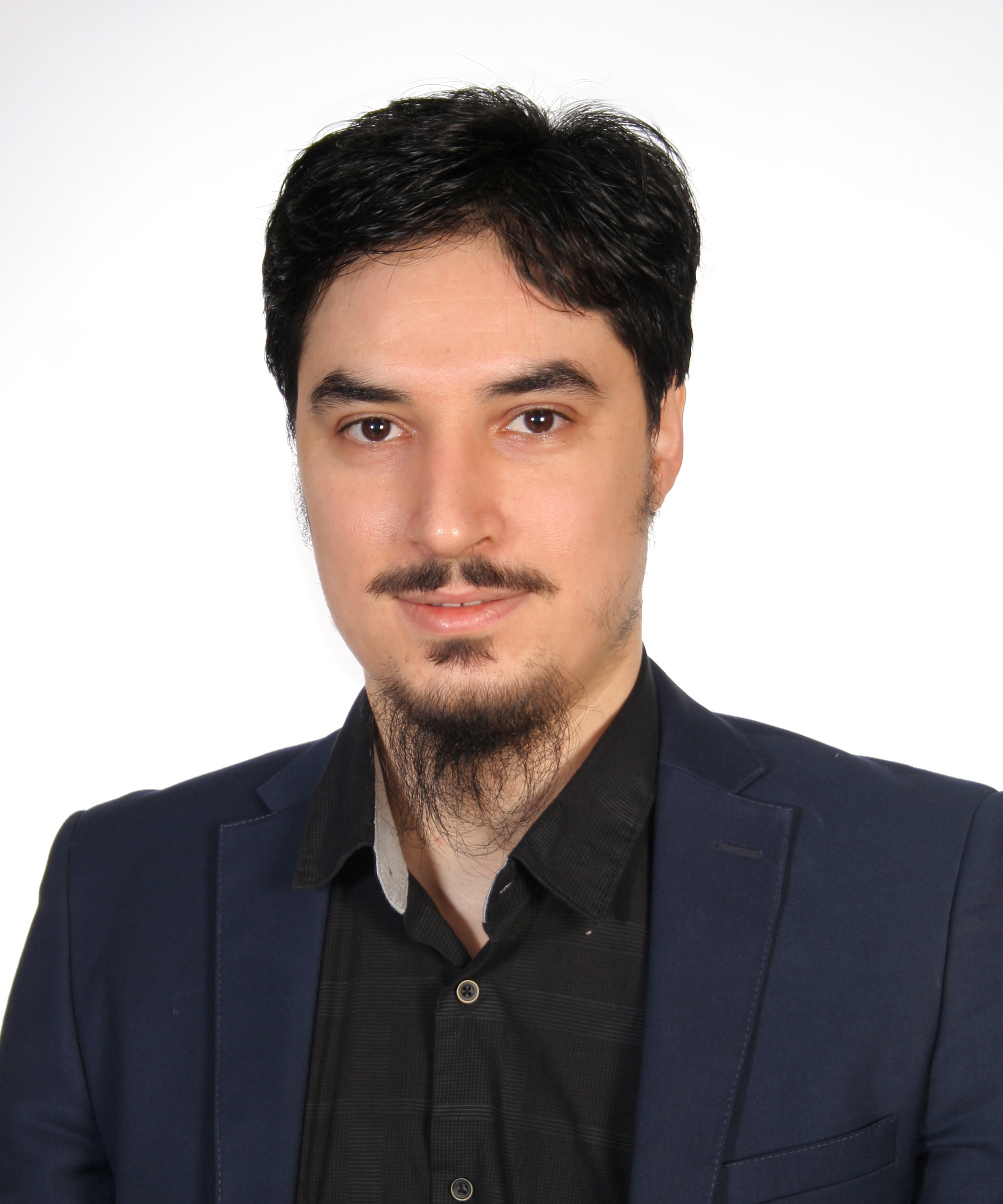}}]{Aymen Khaleel {\normalfont received the B.Sc. degree from the University of Anbar, Al Anbar, Iraq, in 2013, and the M.Sc. degree from Turkish Aeronautical Association University, Ankara, Turkey, in 2017. He is currently pursuing his Ph.D. in Electrical and Electronics Engineering at Ko\c{c} University, Istanbul, Turkey, where he is currently a Project Engineer. His research interests include MIMO systems, index modulation, reconfigurable intelligent surfaces-based systems. He serves as a Reviewer for  \textit{IEEE Transactions on Wireless Communications}, \textit{IEEE Transactions on Vehicular Technology}, \textit{IEEE Communications Magazine}, \textit{IEEE Wireless Communications Letters}, and \textit{IEEE Communications Letters}.}}
\end{IEEEbiography}
%%%%Dr. Ertugrul
\begin{IEEEbiography}[{\includegraphics[width=1in,height=1.25in,clip,keepaspectratio]{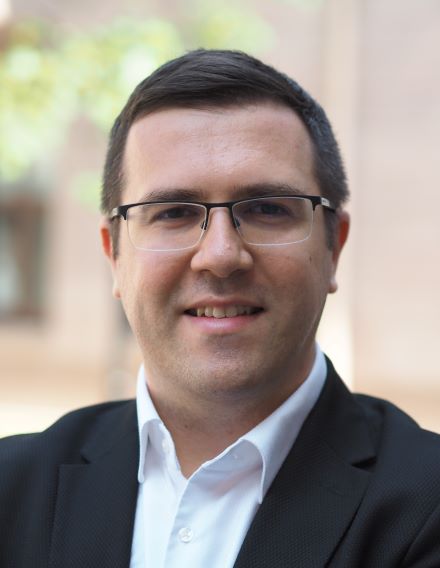}}]{Ertugrul Basar {\normalfont received his Ph.D. degree from Istanbul Technical University in 2013. He is currently an Associate Professor with the Department of Electrical and Electronics Engineering, Ko\c{c} University, Istanbul, Turkey and the director of Communications Research and Innovation Laboratory (CoreLab). His primary research interests include beyond 5G systems, index modulation, intelligent surfaces, waveform design, and signal processing for communications. Dr. Basar currently serves as a Senior Editor of \textit{IEEE Communications Letters} and an Editor of \textit{IEEE Transactions on Communications} and \textit{Frontiers in Communications and Networks}. He is a Young Member of Turkish Academy of Sciences and a Senior Member of IEEE.}}
	
\end{IEEEbiography}
\end{document}